\documentclass[10pt,twocolumn]{article}
\usepackage[a4paper,top=2cm,bottom=2cm, left=1.5cm, right=1.5cm]{geometry}

\usepackage[dvipsnames]{xcolor}
\usepackage[colorlinks,linkcolor=MidnightBlue,citecolor=MidnightBlue,urlcolor=MidnightBlue,linktoc=all]{hyperref}

\usepackage{amssymb}
\usepackage{graphicx}
\usepackage{xurl}
\usepackage{url}
\usepackage{amsmath}
\usepackage{bm}
\usepackage{amsthm}
\usepackage{booktabs}
\usepackage{makecell}
\usepackage{multirow}
\usepackage{comment}
\usepackage{enumerate}
\usepackage{caption}
\usepackage{subcaption}
\usepackage{centernot}
\usepackage{orcidlink}
\usepackage{authblk}
\usepackage{placeins}

\usepackage{enumitem}
\setlist[itemize,1]{label=$\circ$}
\setlist[enumerate,1]{label=(\roman*)}

\usepackage{microtype}
\usepackage{titlesec}
\titleformat{\subsection}
  {\normalfont\normalsize\bfseries}
  {\thesubsection}
  {0.7em}
  {}
\titleformat{\subsubsection}
{\normalfont\normalsize\bfseries}
{\thesubsubsection}
{0.7em}
{}

\usepackage[bibstyle=numeric, citestyle=numeric-comp, sorting=none, isbn=false,url=false,eprint=false]{biblatex}
\DeclareFieldFormat[article]{volume}{\mkbibbold{#1}}
\DeclareFieldFormat[article]{number}{#1}
\renewbibmacro*{volume+number+eid}{%
  \printfield{volume}%
  \setunit*{\addspace}%
  \printfield[parens]{number}%
  \setunit{\addcomma\space}%
  \printfield{eid}}

\renewcommand{\vec}[1]{\mathbf{#1}}

\newcommand{\divergence}[4]{d^{#1}_{#2}\left(#3 \Vert #4\right)}
\newcommand{\divergencename}[2]{d^{#1}_{#2}}

\newcommand{\entropy}[3]{\mathcal{H}^{#1}_{#2}\left(#3\right)}
\newcommand{\entropyname}[2]{\mathcal{H}^{#1}_{#2}}

\newtheorem{definition}{Definition}
\newtheorem{proposition}{Proposition}

\newtheorem{lemma}{Lemma}
\newtheorem{theorem}{Theorem}
\newtheorem{conjecture}{Conjecture}
\usepackage[capitalise]{cleveref}
\crefdefaultlabelformat{#2#1#3}
\creflabelformat{equation}{#2#1#3}
\crefname{definition}{Def.}{Defs.}
\Crefname{definition}{Definition}{Definitions}
\crefname{conjecture}{Conjecture}{Conjectures}
\crefname{corollary}{Corollary}{Corollaries}
\crefname{appendix}{Appendix}{Appendices}

\title{Parametric and structure-aware information theory\\for multi-scale analysis of composition}

\author[1,2,$\dagger$]{Rohit Sahasrabuddhe\,\orcidlink{0000-0002-2779-8310}}
\author[1,2,3,$\ddagger$]{Renaud Lambiotte\,\orcidlink{0000-0002-0583-4595}}
\affil[1]{~Mathematical Institute, University of Oxford, United Kingdom}
\affil[2]{~Institute for New Economic Thinking, University of Oxford, United Kingdom}
\affil[3]{~Complexity Science Hub Vienna, Austria}
\affil[$\dagger$]{~rohit.sahasrabuddhe@maths.ox.ac.uk}
\affil[$\ddagger$]{~renaud.lambiotte@maths.ox.ac.uk}

\usepackage[switch,mathlines]{lineno}

\begin{document}

\twocolumn[
\begin{@twocolumnfalse}

\maketitle
\begin{abstract}
    Compositional data is common across the natural and social sciences, requiring methods to measure the diversity of compositions and the heterogeneity of collections of them. Via Bregman geometry, a strictly concave diversity index induces a measure of heterogeneity that decomposes across scales. A canonical example is Shannon entropy inducing mutual information. We develop this framework by incorporating category similarity into $\alpha$-logarithmic entropy, a parametric diversity index. We disprove an existing general concavity result and prove that, for $\alpha=3$, the region of strict concavity is always convex, complementing previous results for $\alpha=1$ and $2$. The resulting geometry provides methods to compare compositions substantially faster than optimal transport. Applied to occupation compositions across England and Wales, they reveal distinct regionalisations based on different notions of occupation relatedness. Applied to ecological data, they recover established patterns of functional and taxonomic $\beta$-diversity, and reveal sensitivity to the diversity parameter.
\end{abstract}

\vspace{1em}

\end{@twocolumnfalse}
]
\section{Introduction}

Many natural and social systems have elements that belong to different categories, such as organisms of different species or workers in different occupations. Their composition -- the relative abundances of species or employment shares across occupations -- is often represented as a probability distribution. This article is motivated by two common and important questions. First, how diverse is a composition? Second, how much does composition vary between locations? From sociology to ecology, researchers have proposed closely related frameworks to answer them together: choose a diversity index for the first and define the second as the difference between the diversity of the overall composition and the mean diversity of the individual locations~\cite{Theil1971, Reardon2002, MacArthur1966, Lande1996}. A classic example of this is Shannon entropy and mutual information~\cite{Shannon1948}.

While the connection is rarely made explicit (see \textcite{Chodrow2017} for an exception), these frameworks are examples of Bregman geometry~\cite{Bregman1967}; the heterogeneity of the system is the Bregman information induced by a strictly concave diversity index. Variance is another familiar Bregman information, and exemplifies one of their key features: for any grouping of compositions, total Bregman information decomposes into the sum of between-group and within-group contributions~\cite{Banerjee2004}. This decomposition is particularly useful when analysing compositions across scales or levels of aggregation.

The Bregman geometry inherits the limitations of the diversity index that generates it. We focus on two that are important across disciplines. First, any single diversity index embodies one model of the relative importance of rare and common categories. Parametric families of indices allow us to vary this viewpoint interpretably, often connecting several commonly used measures. For instance, the Hill numbers~\cite{Hill1973} include species richness (the number of categories) and transformations of Shannon entropy~\cite{Shannon1948}, the Gini-Simpson index, and the Berger-Parker index~\cite{Berger1970}. Parametricity is useful in many disciplines: $\alpha$-logarithmic entropy, which we use in this article, was proposed by \textcite{Havrda1967} in coding theory and subsequently by \textcite{Patil1982} in ecology and \textcite{Tsallis1988} in statistical mechanics~\footnote{We are using the name from \textcite{Leinster2021}.}. We can vary $\alpha$ to check how sensitive conclusions are to a particular viewpoint on diversity.

The second limitation is that most conventional diversity indices assume that the categories, such as species or occupations, are completely unrelated. Category similarity is crucial in a variety of disciplines. For instance, in evolutionary economic geography, the principle of relatedness posits that economic agents tend to expand their capability base by acquiring new capabilities that are similar to their existing ones~\cite{Hidalgo2018}. It is used to explain phenomena across scales: from the export baskets of countries~\cite{Hidalgo2007}, to how industries enter and exit regions~\cite{Frenken2007, Neffke2011}, and the productivity and diversification of cities~\cite{Bettencourt2014}. In ecology, to paraphrase \textcite{Leinster2012}, disregarding species similarity makes a community of six dramatically different species no more diverse than one with six species of barnacle. Modelling similarity is particularly important in functional ecology, which focuses on organisms' ecological roles rather than their taxonomic classification~\cite{Daz2001, Tilman2001, Violle2007, McGill2006}.

There are several parametric and structure-aware diversity indices addressing these limitations~\cite{Leinster2012}. However, they do not automatically induce Bregman geometry; to do so, we need the index to be strictly concave. Our starting point is the structure-aware generalisation of $\alpha$-logarithmic entropy proposed by \textcite{Ricotta2006}. We find that their result on its concavity is incorrect and provide a counterexample. Adding to existing results for $\alpha=1$ and $2$, we prove that the region of strict concavity for $\alpha=3$ is convex, allowing us to define Bregman geometry on it. This expands the range of viewpoints on diversity that we can capture. Our framework is much faster than optimal transport (OT)~\cite{Villani2009, Peyr2019}, the conventional toolkit for making structure-aware comparisons, providing an efficient alternative when the transport formulation has no substantive interpretation. Applied to the geographic distribution of occupations in England and Wales, our methods reveal how different definitions of relatedness lead to distinct regionalisations. Analysing species abundance data from the Rutor glacier, we recover insights on $\beta$-diversity previously established with OT-based methods.

\section{Results}

\paragraph{Notation.} Throughout this article, we will denote vectors and matrices in bold and their entries in regular font. Unless stated otherwise, powers, logarithms, and comparisons with scalars are element-wise. $\vec{1}$ is the vector of ones and $\vec{I}$ the identity matrix; their dimension will be clear in context. We will use $n$ and $m$ for integers.

We model a composition of $n$ categories as a \textbf{probability distribution} $\vec{p} \in \mathbb{R}^n$ such that $\vec{p} \geq 0$ and $\vec{1}^\top \vec{p} = 1$. The \textbf{probability simplex} $\Delta_n \subset \mathbb{R}^n$ is the set of all distributions over $n$ categories. $\Delta_n^\circ$ is its relative interior, containing all the distributions of full support, i.e. $\vec{p}>0$.

\subsection{Bregman geometry}
As a running example, consider ecological communities composed of $n$ species across $m$ sampling sites. The composition of a site $a$ is $\vec{p}^a \in \Delta_n$, where $p^a_i$ is the relative abundance of species $i$ in it. We compile the relative abundances of the sites themselves into $\vec{w} \in \Delta_m$ such that the overall composition is $\bm{\mu} = \sum_{a} w_a \vec{p}^a$.

To measure biodiversity, we use a \textbf{diversity index} $\mathcal{D}:\Delta_n \to [0,\infty)$. Following \textcite{Whittaker1960}, ecologists have proposed viewing the total diversity as the sum of the diversity of each site and the heterogeneity in composition across sites~\cite{MacArthur1966, Lande1996}. The latter is termed $\beta$-diversity and is given by
\begin{equation}
    \mathcal{D}(\bm{\mu}) - \sum_{a=1}^m w_a\,\mathcal{D}(\vec{p}^a).
\label{eq:jensen_gap_formulation}
\end{equation}
This framework~\footnote{An alternative `multiplicative' partitioning defines $\beta$-diversity as the ratio of the two. The additive version has the advantage that $\beta$-diversity is expressed in the same units. The differing interpretations have been discussed extensively (see \textcite{Veech2010} and the references within).} is used with diversity indices that are strictly concave, guaranteeing that $\beta$-diversity is non-negative. The formulation is not unique to ecology. For example, many segregation indices are also defined as the difference between the diversity of the population and the diversity of the neighbourhoods~\cite{Reardon2002, Chodrow2017}.

In the language of information geometry~\cite{Amari2016}, \cref{eq:jensen_gap_formulation} is the Bregman information induced by $\mathcal{D}$~\footnote{Canonically, Bregman geometry is generated from strictly \textit{convex} functions. We flip signs for equivalent definitions that use strictly concave functions.}. It can also be derived from another route. Assuming that $\mathcal{D}$ is differentiable in $\Delta_n^\circ$, the associated \textbf{Bregman divergence}~\cite{Bregman1967} is
\begin{equation}
    \begin{split}
        d : \Delta_n \times \Delta_n^\circ & \to [0,\infty) \\
         (\vec{p}, \vec{q}) & \mapsto \divergence{}{}{\vec{p}}{\vec{q}}\\
         & \quad = \mathcal{D}(\vec{q}) + \langle \nabla \mathcal{D}(\vec{q}),\, \vec{p}-\vec{q} \rangle - \mathcal{D}(\vec{p}),
    \end{split}
\label{eq:simplex_bregman_divergence}
\end{equation}
where $\langle \cdot, \cdot\rangle$ is the standard inner product. In other words, $\divergence{}{}{\vec{p}}{\vec{q}}$ is how much the first-order approximation from $\vec{q}$ over-estimates $\mathcal{D}(\vec{p})$. The \textbf{Bregman information} induced by $\mathcal{D}$ can be seen equivalently as follows~\cite{Banerjee2004,Chodrow2025}.
\begin{align}
     I : \Delta_n^m \times \Delta_m & \to \left[0, \mathcal{D}(\bm{\mu})\right] \nonumber\\
      \left(\left(\vec{p}^a\right)_{a=1}^m,\,\vec{w}\right) &\mapsto I\left( \left(\vec{p}^a\right)_{a=1}^m;\,\vec{w} \right) \nonumber \\
      &\quad = \mathcal{D}(\bm{\mu}) - \sum_{a=1}^m w_a\,\mathcal{D}(\vec{p}^a) \label{eq:jensen_gap_formulation_defn} \\
      &\quad= \sum_{a=1}^m w_a\, \divergence{}{}{\vec{p}^a}{\bm{\mu}} \label{eq:mean_divergence_formulation}
\end{align}
We omit the explicit dependence of $I$ and $d$ on $\mathcal{D}$ for convenience. \cref{eq:mean_divergence_formulation} provides the interpretation that its divergence from the mean is the `atypicality' of $a$, and heterogeneity is the mean atypicality.

Measuring heterogeneity using Bregman information is particularly compelling because it decomposes at intermediate scales of aggregation. Consider a partition of the compositions into $k$ subsets or groups $\left\{ P_\beta\right\}_{\beta=1}^k$. We write $a \in \beta$ to indicate $\vec{p}^a \in P_\beta$ and define group weights
\begin{equation}
    \vec{\tilde{w}} = \left( \sum_{a \in \beta}w_a \right)_{\beta=1}^k \in \Delta_k,
\end{equation}
within-group weights
\begin{equation}
    \vec{w}^\beta = \left( \frac{w_a}{\tilde{w}_\beta}  \right)_{a\in\beta} \in \Delta_{\vert P_\beta \vert},
\end{equation}
and group means
\begin{equation}
    \bm{\mu}^\beta = \sum_{a \in \beta} w^\beta_a \, \vec{p^a}\in \Delta_n.
\end{equation}
The total heterogeneity decomposes into the amount explained by the partition and the amount within each group~\cite[Theorem 1]{Banerjee2004}.
\begin{equation}
    \begin{split}
        I\left( \left(\vec{p}^a\right)_{a=1}^m ; \, \vec{w} \right) &= \underbrace{I \left( \left( \bm{\mu}^\beta \right)_{\beta=1}^k ;\, \vec{\tilde{w}} \right)}_{\textnormal{between-group}} \\
        & \quad + \underbrace{ \sum_{\beta=1}^k \tilde{w}_\beta \, \underbrace{I\left(\left( \vec{p}^a \right)_{a\in\beta} ;\, \vec{w}^\beta \right)}_{\textnormal{within-}P_\beta} }_{\textnormal{total within-group}}
    \end{split}
\label{eq:decomposition}
\end{equation}

The decomposition is perhaps best-known as a property of variance, which is also a Bregman information. It is valuable, for instance, for measuring how different regions contribute to the total segregation~\cite{Reardon2002, Mora2011}. We can also construct partitions that `explain' a large fraction of heterogeneity using $k$-means-style clustering algorithms~\cite{Banerjee2004} (see \textcite{Chodrow2017} for an example in segregation).

\subsection{The limitations of Shannon information theory}

Shannon information theory provides a paradigmatic example of the Bregman geometry framework~\cite{Cover2005}. \textbf{Shannon entropy}
\begin{equation}
    \begin{split}
        \entropyname{}{}:\Delta_n &\to [0, \, \ln n ] \\
        \vec{p} &\mapsto  -\vec{p}^\top \ln \vec{p} \textnormal{, where } 0\ln0 := 0
    \end{split}
\end{equation}
induces KL divergence and mutual information as its Bregman divergence and information. Shannon entropy is ubiquitous across the natural and social sciences as a diversity index. However, it is not without limitations. We focus on two that are discussed extensively, particularly in ecology. 

First, like any single diversity index, Shannon entropy is only one viewpoint on diversity. A useful interpretation of it in this context is as expected `surprise', where we map the probability of a category to the surprise of observing it with $p \mapsto \ln(1/p)$. This surprise function is not especially meaningful in many applications, and researchers have advocated for \textbf{parametric} families of diversity indices that model several viewpoints on diversity~\cite{Hill1973, Patil1982}.

Second, Shannon entropy is based on the assumption that the categories are completely distinct. This assumption is ill-suited to a wide range of fields, from economic geography~\cite{Hidalgo2018} to ecology~\cite{Leinster2012}, where the similarities between categories are crucial. To capture the intuition that compositions of dissimilar categories are more diverse, we need a \textbf{structure-aware} index.

\textcite{Leinster2012} bring these properties together and discuss structure-aware generalisations of three parametric diversity indices: the Hill numbers~\cite{Hill1973}, R\'enyi entropy~\cite{Renyi1961}, and $\alpha$-logarithmic entropy~\cite{Havrda1967, Patil1982, Tsallis1988}. To construct Bregman geometry that inherits these properties, we need the diversity index to be \textbf{strictly concave}. Even without structure-awareness, only $\alpha$-logarithmic entropy satisfies this for the whole parameter range (\cref{sec:entropy_is_concave}).

\subsection{Structure-aware $\alpha$-logarithmic entropy}

For $\alpha \geq 0$, $\alpha$-logarithmic entropy is defined by generalising Shannon entropy's surprise function to $p \mapsto \ln_\alpha (1/p)$, where $\ln_\alpha$ is the Box-Cox transformation~\cite{Box1964} or $q$-logarithm~\cite{Tsallis1988} (we rename $q$ to $\alpha$ to avoid confusing notation). To encode category similarity, we define an $n \times n$ \textbf{similarity matrix} $\vec{Z}$. We are agnostic of how similarity is defined, as long as it is symmetric ($\vec{Z}=\vec{Z}^\top$), positive and bounded ($Z_{ij} \in [0,1]$), and each category is maximally similar to itself ($Z_{ii}=1$). This similarity structure is a generalisation of coarse-graining (\cref{sec:coarse-graining_limit}), and the assumption of completely distinct categories corresponds to $\vec{Z}=\vec{I}$. The \textbf{ordinariness} of $i$ is $(\vec{Zp})_i = \sum_{j}Z_{ij}p_j \in [0,1]$, its expected similarity to a category picked at random.

We define \textbf{structure-aware $\alpha$-logarithmic entropy} by mapping surprise from ordinariness. For $\alpha \geq 0$,
\begin{align}
    \entropyname{\vec{Z}}{\alpha} : \Delta_n & \to \left[0, \infty \right) \nonumber \\
    \vec{p} & \mapsto \begin{cases}
        \frac{1}{\alpha - 1} \left( 1 - \vec{p}^\top (\vec{Zp})^{\alpha -1} \right) & \textnormal{for } \alpha \neq 1 \\[4pt]
        -\vec{p}^\top \ln \left( \vec{Zp} \right) & \textnormal{for } \alpha = 1,
    \end{cases}  
\end{align}
where $0^\alpha:=0$ and $0\ln0:=0$. $\entropyname{\vec{Z}}{\alpha}$ was proposed by \textcite{Ricotta2006} to bridge between Shannon entropy $\entropyname{\vec{I}}{1}$ and Rao's entropy $\entropyname{\vec{Z}}{2}$~\cite{Rao1982}. Increasing $\alpha$ decreases the surprise of every category and affects the less ordinary ones more (\cref{sec:role_of_alpha}). Thus, it tunes the sensitivity to less ordinary categories~\cite{Leinster2012}.

While $\entropyname{\vec{Z}}{\alpha}$ is strictly concave for $\vec{Z}=\vec{I}$ (\cref{sec:entropy_is_concave}), modelling category similarity destroys this in general. \textcite[Theorem 2]{Ricotta2006} claimed that $\entropyname{\vec{Z}}{\alpha}$ for $\alpha \geq 2$ is concave if $\vec{Z} \succeq 0$. Unfortunately, this claim is false, as we show in \cref{sec:counterexample_ricotta}~\footnote{The error is acknowledged in \textcite{Ricotta2021a}, which the authors shared with us in personal communication. No correction to the original article~\cite{Ricotta2006} has been published.}; $\entropyname{\vec{Z}}{\alpha}$ does not provide a general route to a parametric and structure-aware framework.

\subsection{Limited parametricity}
\begin{figure*}[h!]
    \centering
    \includegraphics{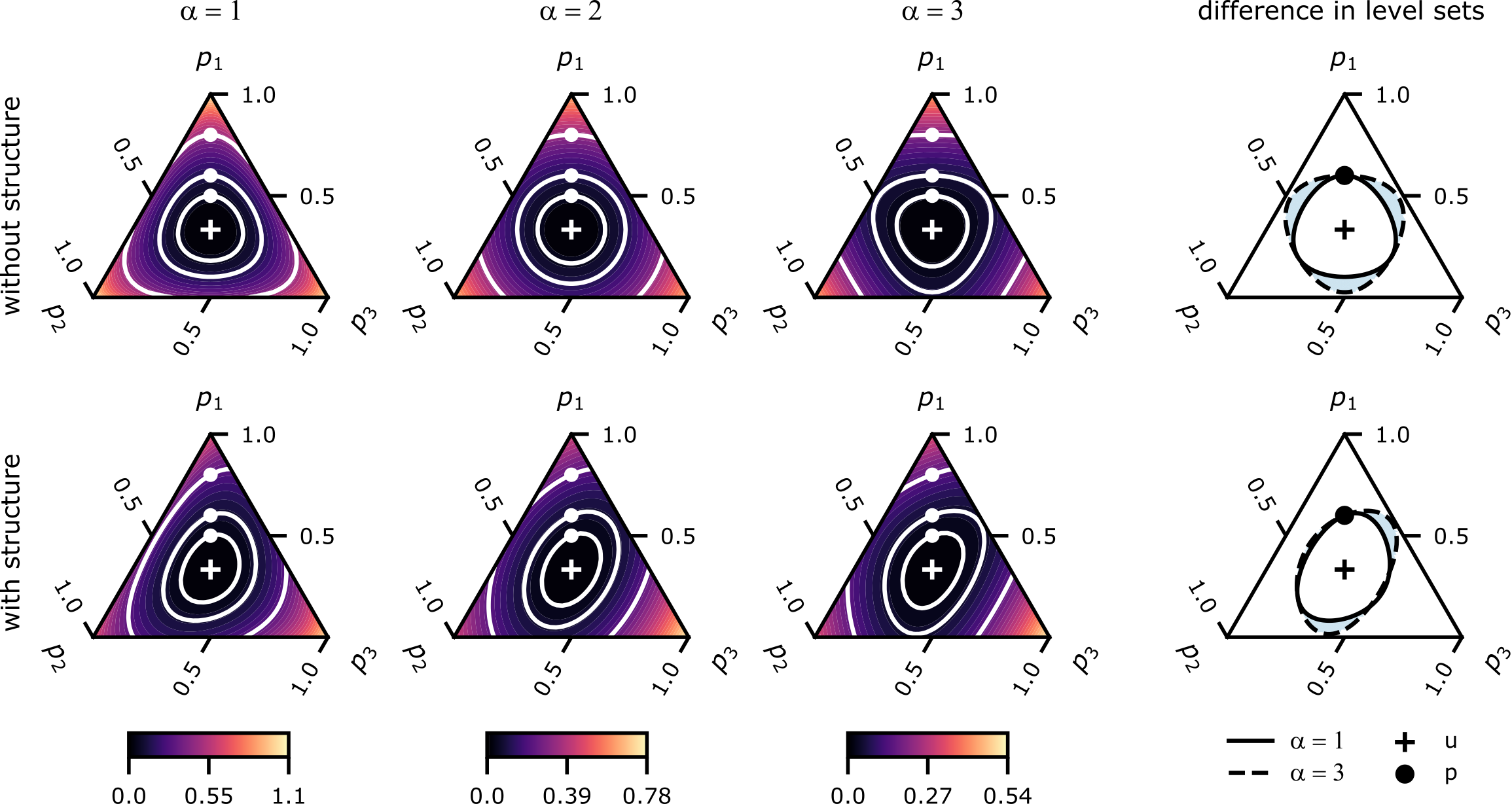}
    \caption{\textbf{Structure-aware divergence.} In the three columns on the left, we plot heatmaps of $\divergence{\vec{Z}}{\alpha}{\vec{p}}{\vec{u}}$ for $\vec{p}\in\Delta_3$, where $\vec{u}$ is the uniform distribution (white cross). The columns correspond to choices of $\alpha$ and each has its own colour scale. The two rows correspond to choices of $\vec{Z}$: $\vec{I}$ in the top row and the first and second categories having similarity $0.5$ in the bottom row. The white lines denote the level sets corresponding to 
    the same three values of $\vec{p}$ in all plots. In the rightmost column, we overlay the level set of the same $\vec{p}$ (black circles) for $\alpha=1$ (solid) and $3$ (dashed). The distributions in the blue region are closer to $\vec{u}$ than $\vec{p}$ is for $\alpha=3$ but not for $\alpha=1$.
    %
    }
    \label{fig:divergence}
\end{figure*}
Fortunately, not all is lost. Researchers have studied the concavity of $\entropyname{\vec{Z}}{\alpha}$ for $\alpha=1$ and $2$. \textcite{Posada2020} conjectured the following (see \textcite[Appendix C.1]{Posada2020} for proof attempts).
\begin{conjecture} (\textcite{Posada2020})
    $\entropyname{\vec{Z}}{1}$ is strictly concave in $\Delta_n^\circ$ if $\vec{Z} \succ 0$ and $\vec{Z}$ satisfies the multiplicative triangle inequality:
    \begin{equation*}
        Z_{ij} \times Z_{jk} \leq  Z_{ik} \textnormal{ for all } 1\leq i,j,k \leq n.
    \end{equation*}
\label{conjecture}
\end{conjecture}
All our results for $\alpha=1$ are conditional on this conjecture. Rao's entropy $\entropy{\vec{Z}}{2}{\vec{p}} = 1 - \vec{p}^\top \vec{Zp}$ and has well-understood behaviour.
\begin{proposition}
    $\entropyname{\vec{Z}}{2}$ is strictly concave in $\Delta_n^\circ$ if $\vec{Z} \succ 0$.
\label{prop:concavity_of_rao}
\end{proposition}

We extend the parameter space by establishing when $\alpha=3$ can be used. The spectrum of its Hessian determines the curvature of $\entropyname{\vec{Z}}{\alpha}$ at each $\vec{p}$. $\entropyname{\vec{Z}}{\alpha}$ is strictly concave `locally' at $\vec{p}$ if $\nabla^2 \entropy{\vec{Z}}{\alpha}{\vec{p}}$ is conditionally negative definite (CND), i.e. negative definite in all directions within the simplex (see \cref{sec:establishing_concavity}). We define the \textbf{region of strict concavity} of $\entropyname{\vec{Z}}{\alpha}$ as
\begin{equation}
    C \left( \entropyname{\vec{Z}}{\alpha} \right) = \left\{ \vec{p}\in\Delta_n^\circ \mid \nabla^2 \entropy{\vec{Z}}{\alpha}{\vec{p}} \textnormal{ is CND} \right\}.
\end{equation}
The region of strict concavity for $\alpha=3$ is a convex set irrespective of $\vec{Z}$ (see \cref{sec:proof} for the proof).
\begin{theorem}
    $C\left(\entropyname{\vec{Z}}{3}\right)$ is convex.
\label{thm:entropy3_convex}
\end{theorem}

\cref{sec:region_of_concavity_nonempty} gives sufficient conditions for $C\left(\entropyname{\vec{Z}}{3}\right)$ to be non-empty, connecting to the theory of magnitude~\cite{Leinster2013}. \cref{thm:entropy3_convex} guarantees that we can define Bregman geometry on $C\left(\entropyname{\vec{Z}}{3}\right)$. We can check if a composition is in $C\left(\entropyname{\vec{Z}}{3}\right)$ by evaluating the spectrum of the Hessian at it (see \cref{sec:establishing_concavity}). If all the compositions in a dataset are in $C\left(\entropyname{\vec{Z}}{3}\right)$, its convexity ensures that their means are as well, and we can perform the decomposition in \cref{eq:decomposition} and Bregman information-based clustering.

\subsubsection{Parametric and structure-aware geometry}
We can construct a Bregman divergence (\cref{eq:simplex_bregman_divergence}) from $\entropyname{\vec{Z}}{\alpha}$ in a convex subset $\mathcal{C} \subset  C\left(\entropyname{\vec{Z}}{\alpha}\right)$. The \textbf{structure-aware divergence} is
\begin{align}
    \divergencename{\vec{Z}}{\alpha}: \mathcal{C} \times \textnormal{ri}\left(\mathcal{C}\right) &\to [0,\infty) \nonumber \\
    (\vec{p}, \vec{q}) & \mapsto \divergence{\vec{Z}}{\alpha}{\vec{p}}{\vec{q}}, \nonumber
\end{align}
where for $\alpha=1$, it generalises KL divergence as below,
\begin{equation}
    \divergence{\vec{Z}}{1}{\vec{p}}{\vec{q}} = 1+\vec{p}^\top \ln\frac{\vec{Zp}}{\vec{Zq}} -\vec{q}^\top \frac{\vec{Zp}}{\vec{Zq}}
\end{equation}
and for $\alpha \neq 1$,
\begin{equation}
    \begin{split}
        \divergence{\vec{Z}}{\alpha}{\vec{p}}{\vec{q}} &= \frac{1}{\alpha-1}\,
            \vec{p}^\top
            \left(
                (\vec{Zp})^{\alpha-1}
                -(\vec{Zq})^{\alpha-1}
            \right) \\
            &\qquad
            -\vec{q}^\top
            \operatorname{diag}\left((\vec{Zq})^{\alpha-2}\right)
            \vec{Z}(\vec{p}-\vec{q}).
    \end{split}
\end{equation}
For $\alpha=2$, it is the squared Mahalanobis distance~\cite{Mahalanobis1936}, which generalises the squared Euclidean distance. The choice of $\alpha$ and $\vec{Z}$ has a clear effect on $\divergencename{\vec{Z}}{\alpha}$ (\cref{fig:divergence}). Moving probability mass between similar categories contributes less to divergence, and judging which of two distributions is closer to a reference can depend on $\alpha$. We define \textbf{structure-aware information} $I^\vec{Z}_\alpha$ following \cref{eq:jensen_gap_formulation_defn,eq:mean_divergence_formulation}.

\subsubsection{Constructing similarity matrices}
To satisfy \cref{conjecture}, $\vec{Z}$ needs to be positive definite and satisfy the multiplicative triangle inequality. Inheriting it from a metric is a convenient way to ensure the latter~\footnote{This is the \textit{only} way for $\vec{Z}$ to satisfy the multiplicative triangle inequality when $Z_{ij} \in (0,1)$ for all $i\neq j$. (\cref{prop:metric_similarity}).}. For a metric $d$ such that $d(i,j)$ is the distance between categories $i$ and $j$, we can define $\vec{Z}$ such that
\begin{equation}
    Z_{ij} = \exp \left( - d(i,j) \right).
\label{eq:z_from_d_tau1}
\end{equation}
The choice of base $e$ is arbitrary and other choices are equivalent to linearly scaling the metric by $\tau > 0$. Thus, we generalise to a family of similarity matrices $\vec{Z}(\tau)$, where
\begin{equation}
    Z(\tau)_{ij} = \exp \left( - \tau \, d(i,j) \right).
\end{equation}

How do we ensure $\vec{Z} \succ 0$? As $\tau \to \infty$, $\vec{Z}(\tau)$ approaches $\vec{I}$: the categories are so far apart that they are almost completely distinct. Since $\vec{I}\succ 0$ and the set of positive definite matrices is open, there is always a $\tau$ such that $\vec{Z}(\tau)$ is \textit{close enough} to $\vec{I}$ to be positive definite (see \textcite[Proposition 10.8]{Leinster2016} for a formal proof). Intuitively, for $\vec{Z}$ to be positive definite, the categories must be \textit{distinct enough}. Scaling the metric such that the minimum distance is more than $\ln (n-1)$ guarantees this (\textcite[Proposition 2.1.3]{Leinster2013}); less extreme scaling will often work in practice.

Another approach is to use metrics of negative type~\cite{Schoenberg1938}. A metric $d$ is of \textbf{negative type} if the distance matrix $\vec{D}$, where $D_{ij} = d(i,j)$, is conditionally negative semidefinite or equivalently, if $\vec{D}^{1/2}$ is isometric to a subset of a Hilbert space. \textcite[Theorem 3.3]{Meckes2013} shows that if $d$ is of negative type, $\vec{Z}(\tau) \succ 0$ for all $\tau > 0$.  Many metric spaces that arise naturally in models of real systems are of negative type -- for instance, subsets of Euclidean space, weighted trees with the shortest path distance~\cite{Hjorth1998}, and undirected graphs with the effective resistance distance~\cite{Fiedler2011, Devriendt2022Thesis}.

\subsection{Comparison with optimal transport}\label{sec:comparison_ot}

Optimal transport (OT) theory provides a suite of tools to compare distributions over metric spaces (see \textcite{Villani2009} and \textcite{Peyr2019} for overviews). Computing OT distances involves solving an optimisation problem. In cases when this problem is not physically meaningful to the application, structure-aware divergence is a closed-form alternative that is substantially faster. Note that the divergences are not approximations of OT distance: they are fundamentally different ways of comparing distributions.

\begin{figure}[h!]
    \centering
    \includegraphics[width=\columnwidth]{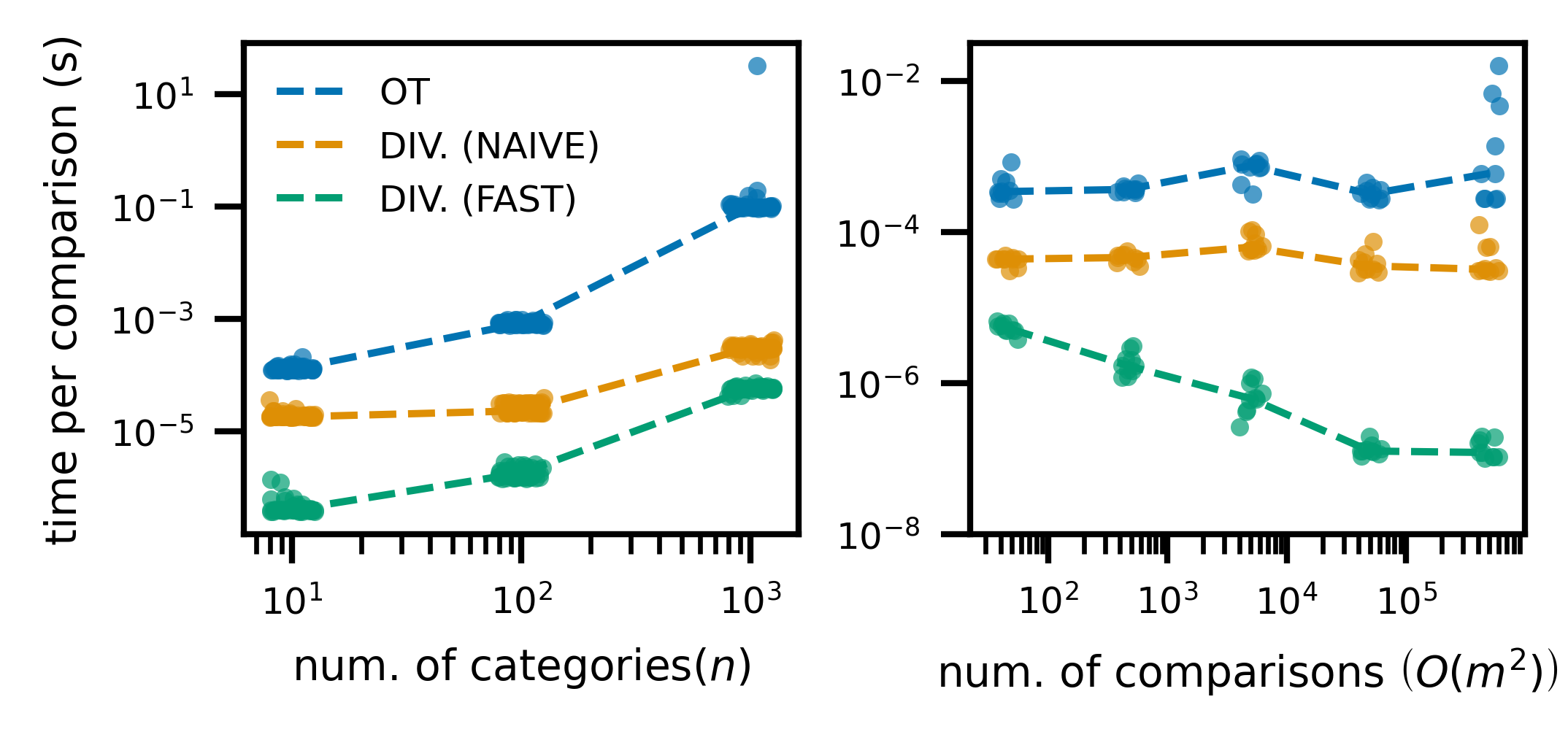}
    \caption{\textbf{Runtime vs optimal transport.} The time per comparison as a function of the number of categories $n$ (left) and the number of comparisons in the all-pairs task (right). The colour denotes the algorithm. The markers correspond to individual runs and the dashed lines connect the median values. Note the log scale on both axes. These results are for $\alpha=2$ and are qualitatively similar to those for $\alpha=1$ and $3$ (\cref{fig:ot_appendix}).}
    \label{fig:ot}
\end{figure}

We conduct two experiments to quantify the difference in runtime. The first investigates the effect of the number of categories $n$. The task is to make 100 pairwise comparisons of randomly sampled compositions, with category similarity also generated at random (see \cref{sec:methods_ot} for details). We implement the following algorithms.
\begin{enumerate}[nosep, ]
    \item {\scriptsize OT}: Iterate over the pairs and compute the Wasserstein distance~\cite{Vaserten1969, Kantorovich1960}.
    \item {\scriptsize DIV.~(NAIVE)}: Iterate over the pairs and compute structure-aware divergence.
    \item {\scriptsize{DIV.~(FAST)}}: Compute all $100$ divergences at once.
\end{enumerate}
We assess each algorithm with the time taken per comparison for $n=10$, $100$, and $1000$, running the task $50$ times for each value. We do not compare against OT techniques such as entropic regularisation~\cite{Cuturi2013} since they involve trade-offs between accuracy and speed that we have no clear way of making.

{\scriptsize DIV.~(NAIVE)} is consistently faster than {\scriptsize OT} (\cref{fig:ot}). The gap widens as $n$ increases to nearly three orders of magnitude for $n=1000$. Vectorising provides further gains, although they diminish for larger $n$.

Some tasks can be vectorised to be dramatically quicker. We demonstrate this with our second experiment: the all-pairs task. Fixing $n=50$, we sample $m$ compositions and category similarity randomly in the same way as before (\cref{sec:methods_ot}). The task is to make all $m \choose 2$ comparisons. {\scriptsize{OT}} outputs the Wasserstein distance, and {\scriptsize DIV.~(NAIVE)} and {\scriptsize DIV.~(FAST)} output the structure-aware information of each pair. {\scriptsize{OT}} and {\scriptsize DIV.~(NAIVE)} iterate through each pair, while {\scriptsize DIV.~(FAST)} is optimised specifically for the all-pairs task~\footnote{Note that {\scriptsize DIV.~(FAST)} is not an approximation; it computes structure-aware information exactly.}. We conduct 10 runs of the task each for five values of $m \in [10,\,10^3]$.

As expected, varying $m$ has no real impact on the per comparison runtime of {\scriptsize OT} and {\scriptsize DIV.~(NAIVE)}: the latter is consistently roughly an order of magnitude faster (\cref{fig:ot}). {\scriptsize DIV.~(FAST)} gets faster per comparison as $m$ increases, with the largest gains for $\alpha=2$, whose quadratic form is particularly efficient to compute (\cref{fig:ot_appendix}).

\subsection{The geography of occupation composition}

\begin{figure*}[h!]
    \centering
    \includegraphics[width=\textwidth]{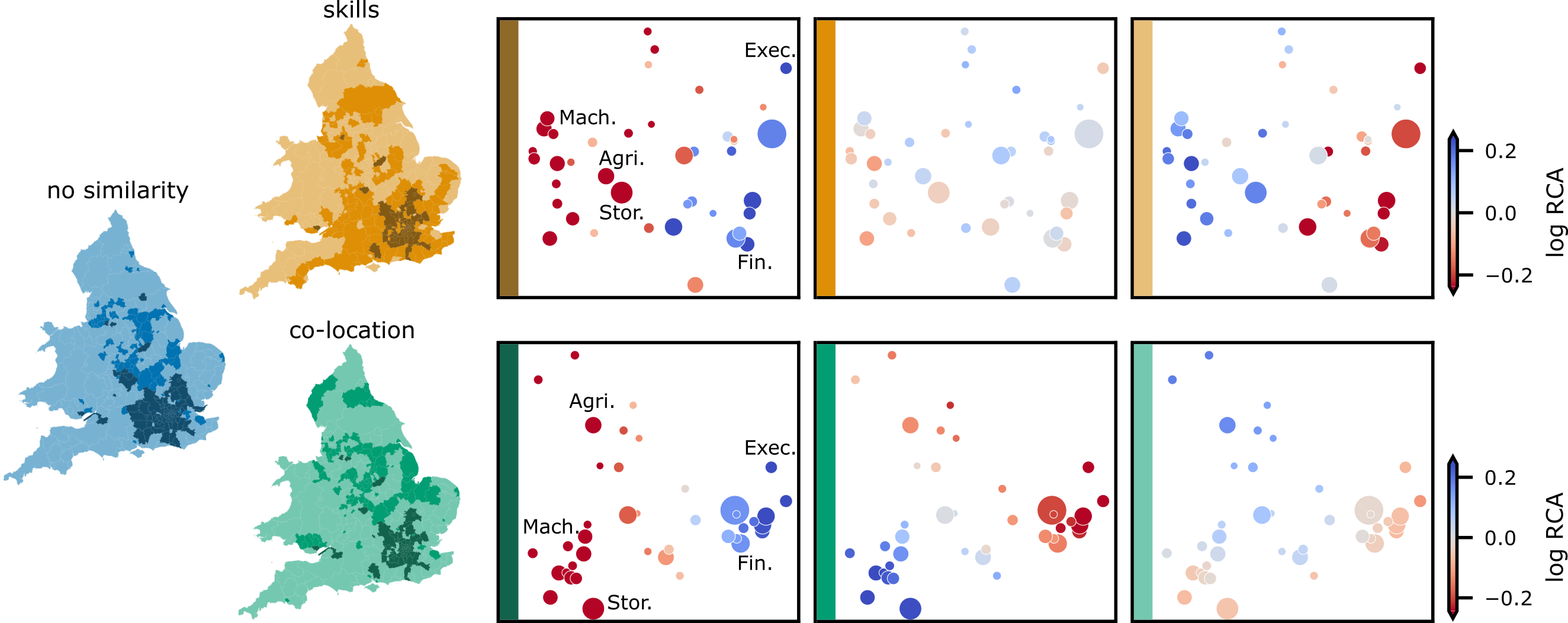}
    \caption{\textbf{Three-way regionalisations by occupation composition.}
    (Left) The geography of the optimal partitions into three regions using no occupation similarity (blue), skills similarity (orange), and co-location similarity (green). We distinguish the regions by intensity of colour (light, medium, and dark). 
    (Right) The occupation composition of each cluster for skills (top row) and co-location (bottom row) similarity. Each point is an occupation, with size corresponding to employment share in the total population. We label a few distinctive occupations (see end of caption for abbreviations). Their positions are set by multidimensional scaling (MDS) of the similarity structure: similar occupations are closer. Each plot corresponds to the region indicated by the stripe on its left margin. The colour of an occupation represents the Revealed Comparative Advantage (RCA) of the corresponding region in it. Blue and red indicate over and under-representation respectively. Abbreviations: Exec. is \textit{Chief Executives and Senior Officials}, Fin. is \textit{Finance Professionals}, Mach. is \textit{Plant and Machine Operatives}, Stor. is \textit{Elementary Storage Occupations}, and Agri. is \textit{Agricultural and Related Trades}.
    }
    \label{fig:ew_maps}
\end{figure*}

A region's occupation composition serves as a proxy for its capability base, which shapes key features such as resilience to shocks and the ability to diversify into new economic activities~\cite{Neffke2011}. \textcite{Daniotti2025} define the expected similarity of the occupations of two randomly sampled workers as the \textit{coherence} of a city, and find large-scale patterns in how it varies with time and population size. In our framework, coherence is $1-\entropyname{\vec{Z}}{2}$. This motivates using $\divergencename{\vec{Z}}{\alpha}$ to compare occupation compositions and $I^\vec{Z}_\alpha$ to describe patterns across spatial scales.

We analyse occupation data from the 2021 census of England and Wales at the geographic scale of Local Authorities (LAs). We proxy each LA's capability base as its employment shares across occupations that produce tradable output, i.e. excluding those such as health service and retail workers that cater mainly to the local population~\cite{Daniotti2025}. We identify 41 such occupations and represent each of the 318 LAs as a composition of them (see \cref{sec:tradable_occupations}).

There are several ways to conceptualise the relatedness of occupations. We construct similarity matrices for two commonly used notions: occupations are related if (i) they require the same skills (skills similarity) or (ii) the same regions specialise in them (co-location similarity) (see \cref{sec:occ_similarity_matrices,fig:occs_Z}). These definitions are fundamentally different: co-located occupations tend to belong to similar sectors, but may require very different skills.

\subsubsection{Regionalisation}
Grouping LAs into regions by occupation composition reveals broad patterns in specialisation. Our analysis of the diversity of the LAs (\cref{sec:diversity_LA_appendix}) suggests that the range of $\alpha=1$ to $3$ is too narrow to capture a substantially wide variety of perspectives in this dataset. Therefore, we focus solely on how the choice of $\vec{Z}$ affects regionalisation. We pick $\alpha=2$ since it aligns with \textcite{Daniotti2025} and employ a $k$-means-style algorithm (see \cref{sec:kmeans} for details).

The two-way partitions are similar for all choices of $\vec{Z}$ and explain a large fraction of heterogeneity (\cref{fig:ew_clustering_stats}). This division captures the concentration of white collar employment in the region surrounding London and in small pockets around cities such as Manchester, Bristol, and Cardiff (\cref{fig:ew_maps_2}). For larger $k$, the optimal partitions can be quite different, as we illustrate for $k=3$ in \cref{fig:ew_maps}.

The cardinal advantage of using Bregman geometry is the ability to decompose total heterogeneity at the region scale (\cref{eq:decomposition}). Separately for either definition of occupation similarity, we write
\begin{equation}
    1 = \underbrace{\sum_{\beta=1}^k \tilde{w}_\beta \, \divergence{\vec{Z}}{2}{\bm{\mu}^\beta}{\bm{\mu}}}_{\textnormal{information explained}} + \underbrace{ \sum_{\beta=1}^k \tilde{w}_\beta \, I\left(\left( \vec{p}^a \right)_{a\in\beta} ;\, \vec{w}^\beta \right) }_{\textnormal{information within}},
\label{eq:ew_decomposition}
\end{equation}
where we have normalised such that the total Bregman information is $1$. The results in \cref{tab:ew_clusters} allow us to describe the regions quantitatively.

\begin{table*}[h!]
    \centering
    \begin{tabular}{c|c|cc|cc}
        \toprule
         \textbf{Region} & \textbf{Pop.} & \textbf{Atypicality} & \multirowcell{2}{\textbf{Information}\\\textbf{explained}} & \textbf{Heterogeneity} & \multirowcell{2}{\textbf{Information}\\\textbf{within}} \\
         $\beta$ & $\tilde{w}_\beta$ & $\divergence{\vec{Z}}{2}{\bm{\mu}^\beta}{\bm{\mu}}$ &  & $I\left(\left( \vec{p}^a \right)_{a\in\beta} ;\, \vec{w}^\beta \right)$  &  \\
         \midrule
         \multicolumn{3}{l}{\textbf{skills similarity}} & 66.3\% &  & 33.7\%\\
         \midrule
        light & 38\% & 0.75  & 28.7\%   & 0.38 & 14.5\% \\
        medium & 37\% & 0.020 & 0.8\% & 0.28 & 10.3\% \\
        dark & 25\% & 1.5   & 36.9\%   & 0.36 & 8.9\% \\
        \midrule
         \multicolumn{3}{l}{\textbf{co-location similarity}} & 67.0\% & & 33.0\%\\
         \midrule
        light & 43\% & 0.060 & 2.6\% & 0.32 & 13.8\% \\
        medium & 29\% & 0.91  & 26.5\%  & 0.30 & 8.9\% \\
        dark & 28\% & 1.4   & 37.9\%  & 0.37 & 10.3\% \\
         \bottomrule
    \end{tabular}
    \caption{\textbf{Decomposing heterogeneity of occupation composition.} The atypicality and heterogeneity of each region in three-way partitions of England and Wales, calculated using \cref{eq:ew_decomposition}. The region names correspond to colour intensities in \cref{fig:ew_maps}. The information explained by region $\beta$ is $\tilde{w}_\beta\, \divergence{\vec{Z}}{2}{\bm{\mu}^\beta}{\bm{\mu}}$ and the information within it is $\tilde{w}_\beta\, I\left(\left( \vec{p}^a \right)_{a\in\beta} ;\, \vec{w}^\beta \right)$. We normalise separately for the two definitions of similarity such that the total Bregman information is $1$. The minor discrepancy in the sums is due to rounding.}
    \label{tab:ew_clusters}
\end{table*}

All three partitions reveal a services-heavy region around London (dark shade, \cref{fig:ew_maps}). This region is similar in the structure-aware partitions, houses around a quarter of the population, and is very atypical of the overall population (\cref{tab:ew_clusters}). It is substantially larger in the structure-blind partition, with more than a third of the population. The other two regions differ markedly depending on the definition of similarity.

Using co-location similarity reveals sector-specialised regions. The industry-heavy region (medium green) includes population centres such as Birmingham, Northamptonshire, and Bradford. The vast agricultural region (light green) spans nearly half the LAs and 43\% of the population, encompassing nearly all of Yorkshire, Cornwall, and Wales. While it is fairly heterogeneous internally, its mean composition is representative of the overall composition.

Under skill similarity, some industrial and agricultural LAs merge into one region (light orange). The overlap in the skills of blue collar occupations across sectors means that LAs with large population shares in them are grouped together. The remaining skill-based region has no clear specialisation (medium orange). Despite having more than a third of the population, it only explains 0.8\% information, indicating that on average, it is very typical of the overall population.

Providing deeper economic insights based on these results is beyond our expertise and the purpose of this article. The main takeaway is that our clustering is sensitive to the similarity structure in intuitive ways and that the decomposition provides useful insights into the clusters. In economic geography, this enables practitioners to explore regionalisations with notions of relatedness tailored to their research questions.

\subsection{Ecological $\beta$-diversity} \label{sec:beta_diversity}

Returning to our motivating example, we apply our methods to a case study of ecological $\beta$-diversity. The perspective of functional ecology prioritises the ecological roles of organisms above their taxonomic classification~\cite{McGill2006,Violle2007}. When taxonomic $\beta$-diversity is high, it is important to know whether the new species occupy the same functional niches in the community. Ecologists often proxy function using traits and define species similarities using the similarities of their traits.

\textcite{Ricotta2021b} proposed a functional $\beta$-diversity measure built on an OT approach for comparing trait distributions~\cite{Kosman1996, Gregorius2003}. Their method matches species between two sites in order to minimise the mean dissimilarity between pairs of matched individuals. This comes at a computational cost, which they remark upon. This is an example of when the OT methods have no physical interpretation, and can be substituted by faster tools.

\subsubsection{Vegetation in the Rutor glacier}

We show that our framework reproduces the insights of \textcite{Ricotta2021b} and \textcite{Ricotta2024} on $\beta$-diversity of vegetation distribution in the Rutor glacier~\cite{Caccianiga2006}. The data consists of the abundances and functional traits of 45 plant species in 59 plots along the glacier. The plots are along a primary succession, where the retreating glacier exposed previously uninhabited soil. They come from three different successional stages; there are 17 early, 32 mid, and 10 late stage plots.

We follow \textcite{Ricotta2021b} exactly to define a species similarity matrix $\vec{Z}$ (\cref{sec:species_similarity_matrix,fig:rutor_Z}). $\vec{Z} \succ 0$, allowing us to use $\alpha=2$ (\cref{prop:concavity_of_rao}). $\entropyname{\vec{Z}}{3}$ is strictly concave at the species compositions of all the sites, opening up $\alpha=3$ (\cref{thm:entropy3_convex}). We do not use $\alpha=1$ because $\vec{Z}$ does not satisfy the multiplicative triangle inequality, falling outside the regime covered by \cref{conjecture}.

\begin{table}[h!]
    \centering
    \begin{tabular}{c|cccc}
    \toprule 
         \multirow{2}{*}{$\beta$\textbf{-diversity}} & \multicolumn{2}{c}{$\alpha=2$} & \multicolumn{2}{c}{$\alpha=3$} \\
         & $\vec{I}$ & $\vec{Z}$ & $\vec{I}$ & $\vec{Z}$ \\
         \midrule
         explained & 0.28 & 0.52 & 0.20 & 0.52 \\
         \midrule
         early & 0.90 & 0.79 & 1.11 & 0.70 \\
         mid & 0.60 & 0.36 & 0.57 & 0.39 \\
         late & 0.77 & 0.36 & 1.02 & 0.40 \\
    \bottomrule
    \end{tabular}
    \caption{\textbf{Decomposition of $\beta$-diversity in the Rutor glacier.} The values in the decomposition for each $\alpha$ and $\vec{Z}$ (see \cref{eq:rutor_decomposition}). $\vec{I}$ models the taxonomic perspective and $\vec{Z}$ the functional.}
    \label{tab:rutor_beta_diversity}
\end{table}

For each choice of $\alpha$ and $\vec{Z}$, we compute $\beta$-diversity as the structure-aware information, where all the plots have equal weight. We decompose $\beta$-diversity at the stage scale as follows, normalising such that the total $\beta$-diversity is $1$ (\cref{eq:decomposition}).
\begin{equation}
    1 = \textnormal{explained} 
        +  \frac{17}{59} \times \textnormal{early}
        +  \frac{32}{59} \times \textnormal{mid}
        +  \frac{10}{59} \times \textnormal{late}
\label{eq:rutor_decomposition}
\end{equation}
We report these values in \cref{tab:rutor_beta_diversity}. To check for significance, we compare them to 1,000 samples from a null model where we randomly partition the plots into stages with the same size distribution: 17, 32, and 10.

In all cases, the stage scale explains significantly more $\beta$-diversity than by chance (\cref{fig:rutor_explained_null}). However, the different stages are not equally cohesive (\cref{fig:rutor_plots_embedding}). Previous studies~\cite{Ricotta2021b, Ricotta2024} found that while species turnover amongst plots is comparable for the three stages, the mid and late stages have much lower functional $\beta$-diversity than the early stage. Our results for $\alpha=2$ support this idea (\cref{tab:rutor_beta_diversity}): the functional $\beta$-diversities of the mid and late stages are less than half that of the early stage. We put these values in the context of their null distributions in \cref{fig:rutor_2}. While all the stages are more cohesive than expected by chance, the $\beta$-diversities of the early stage are well within the null distributions. The mid stage is significantly cohesive both taxonomically and functionally, and the late stage only functionally.

\begin{figure}[h!]
    \centering
    \includegraphics[width=\columnwidth]{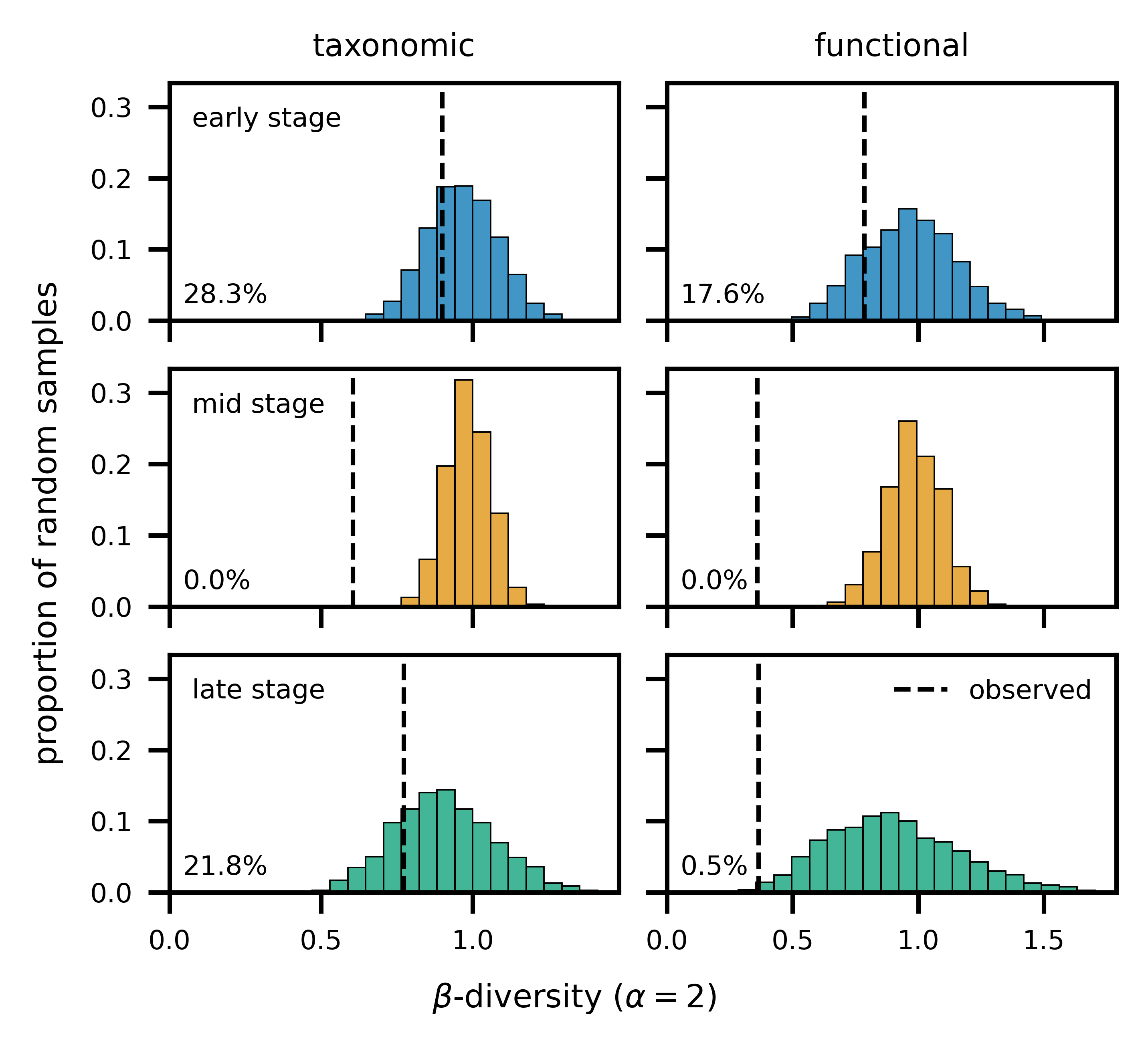}
    \caption{\textbf{$\beta$-diversity of the successional stages of the Rutor glacier.} We plot the taxonomic (left) and functional (right) $\beta$-diversity of the three successional stages (rows). The distributions are the $\beta$-diversities of 1,000 draws from a null model where the partition into stages is random. The dashed black lines mark the empirical values (\cref{tab:rutor_beta_diversity}). We annotate the bottom left corner of each subplot with the percentage of random samples with $\beta$-diversity lower than the empirical. These results are for $\alpha=2$; see \cref{fig:rutor_3} for $\alpha=3$.}
    \label{fig:rutor_2}
\end{figure}

The headline ecological conclusion is qualitatively unchanged when we reduce the importance of species with rarer traits ($\alpha=3$ in \cref{tab:rutor_beta_diversity}): the mid and late stages are much more functionally cohesive than the early stage. However, comparisons with the null model change quite dramatically (\cref{fig:rutor_3}). For instance, the early and late stages are now less taxonomically cohesive than expected by chance, albeit not significantly. Results also move in the opposite direction: the early stage's functional cohesion becomes more significant.

We highlight two main takeaways from this analysis. First, our method for measuring functional $\beta$-diversity reveals similar patterns as existing OT methods, while making it computationally feasible to simulate null models. Second, our results are quantitatively sensitive to $\alpha$, suggesting that the sensitivity to rare species or traits matters. In this data, we capture substantial variation even with $\alpha \in \{2,3\}$.

\section{Discussion}

Our aim in this article is to develop a framework for analysing compositions across scales that captures similarities between categories and different viewpoints on diversity. We do so by constructing Bregman geometry from a diversity index that is parametric, structure-aware, and strictly concave. Established indices satisfy at most two of these three properties: $\alpha$-logarithmic entropy~\cite{Havrda1967, Patil1982, Tsallis1988} is parametric and strictly concave, Rao's entropy~\cite{Rao1982} is structure-aware and strictly concave for $\vec{Z}\succ0$, and there are several parametric and structure-aware indices~\cite{Leinster2012}. To combine all three, we study the structure-aware $\alpha$-logarithmic entropy $\entropyname{\vec{Z}}{\alpha}$~\cite{Ricotta2006}.

Although we disprove an existing concavity result for $\alpha \geq 2$, we argue that $\entropyname{\vec{Z}}{\alpha}$ is still useful over a restricted range of $\alpha$. We show that category similarities generated from metric spaces of negative type~\cite{Schoenberg1938, Meckes2013} satisfy the conditions conjectured for $\entropyname{\vec{Z}}{1}$ and established for $\entropyname{\vec{Z}}{2}$ to be strictly concave. In \cref{thm:entropy3_convex}, we prove that the region of strict concavity of $\entropyname{\vec{Z}}{3}$ is always convex. Thus, when the compositions in the data lie within it, we can construct Bregman geometry. We have not found conditions that guarantee they do, but no composition fell outside this region across all our experiments, suggesting that in practice, $\alpha=3$ is often usable alongside $\alpha=1$ and $2$.

The Bregman divergence and information induced by $\entropyname{\vec{Z}}{\alpha}$ have closed-form expressions and, in our experiments, are substantially faster to compute than exact OT distances. This makes them useful for comparing distributions when the transport formulation is not of interest. In a case study of ecological diversity, we recover insights previously established with OT methods. Scalability will be increasingly important as datasets of species abundances and traits grow larger~\cite{Jones2009,Tobias2022}. In the same experiment, comparisons against a null model are sensitive to $\alpha$, underlining the value of parametricity.

We demonstrate the utility of Bregman geometry-based clustering~\cite{Banerjee2004} in a case study from economic geography, where we regionalise based on occupation composition. The definition of occupation similarity shapes the results substantially and in interpretable ways; practitioners can choose a notion of relatedness tailored to their research question. Such data-driven regionalisation is widely applicable. For instance, it allows researchers to cope with the modifiable areal unit problem in studies of segregation~\cite{Chodrow2017}.

While our framework accommodates a wide range of similarity structures, its main limitation is the narrow range of admissible $\alpha$. When similarity has a group structure, such as when it is derived from a hierarchical classification, we can work around this limitation by coarse-graining the categories. Treating the categories at each level as distinct, we can construct Bregman geometry from $\entropyname{\vec{I}}{\alpha}$ for any $\alpha>0$ and study the system at each level separately. This strategy can be viewed as constraining similarity structure to open up the entire range of $\alpha$. Investigating the tradeoff between flexibility in $\vec{Z}$ and in $\alpha$, including by starting from other diversity indices, is a useful direction for future research. Another avenue is whether multiplicative diversity partitioning~\cite{Veech2010,Reeve2014} offers a practical route to multi-scale analysis without requiring strict concavity. Within our framework, developing statistical estimators for the quantities we introduce is a critical next step. Sampling intensity may influence the choice of $\alpha$ in empirical studies: reducing sensitivity to rare categories by increasing $\alpha$ may improve robustness when those categories are poorly sampled.

In summary, this work brings parametricity and category similarity together in a framework for decomposing compositional heterogeneity across scales. Although this comes with restrictions, our results open a route towards more flexible analyses of composition.

\section{Methods}

\subsection{Establishing strict concavity} \label{sec:establishing_concavity}

$\entropyname{\vec{Z}}{\alpha}$ is twice-differentiable in $\Delta_n^\circ$, allowing us to use the second-order conditions for concavity~\cite[Section 3.1.4]{Boyd2004}.
\begin{lemma}
    $\entropyname{\vec{Z}}{\alpha}$ is locally strictly concave at $\vec{p}\in \Delta_n^\circ$ if its Hessian evaluated at $\vec{p}$ is negative definite.
\label{lemma:second_order}
\end{lemma}

To write the Hessian, we define
\begin{align}
    \vec{D_1}^{\vec{p},\alpha} &= \textnormal{diag}\left(\left(\vec{Zp}\right)^{\alpha-2}\right) \textnormal{ and}\\
    \vec{D_2}^{\vec{p},\alpha} &= \textnormal{diag}\left(\vec{p} \otimes \left(\vec{Zp}\right)^{\alpha-3}\right),
\end{align}
where $\otimes$ is element-wise multiplication. The gradient and the Hessian are as follows.
\begin{align}
    \nabla \entropy{\vec{Z}}{\alpha}{\vec{p}} &= - \frac{1}{\alpha - 1} \left(\vec{Zp}\right)^{\alpha-1} - \vec{Z} \, \vec{D_1}^{\vec{p},\alpha} \, \vec{p} \\
    \nabla^2 \entropy{\vec{Z}}{\alpha}{\vec{p}} &= - \left( \vec{Z}\,\vec{D_1}^{\vec{p},\alpha} + \vec{D_1}^{\vec{p},\alpha}\,\vec{Z} \right. \nonumber\\ 
    & \left. \qquad\quad+  (\alpha-2)\; \vec{Z}\,\vec{D_2}^{\vec{p},\alpha}\,\vec{Z} \right)
\label{eq:hessian}
\end{align}

The spectrum of the Hessian evaluates curvature in all directions, whereas we only need strict concavity to hold in `valid' directions: those within $\Delta_n^\circ$. This set of directions is the tangent space
\begin{equation}
    T\Delta_n = \left\{ \vec{x} \in \mathbb{R}^n \mid \vec{1}^\top \vec{x}=0 \right\}.
\label{eq:tangent_space}
\end{equation}
What we need is for $\nabla^2 \entropy{\vec{Z}}{\alpha}{\vec{p}}$ to be conditionally negative definite, defined as follows~\footnote{Our terminology differs deliberately from the convention, in which CND does not guarantee that the inequality is strict. The convention of `CND' and `strict CND' is ``\textit{unfortunately well-established}"~\cite{Meckes2013}}.
\begin{definition}
    A real symmetric matrix $\vec{M} \in \mathbb{R}^{n\times n}$ is \textbf{conditionally negative definite (CND)} if
    \begin{equation}
        \vec{x}^\top \, \vec{M} \, \vec{x} < 0 \textnormal{ for all } \vec{x} \in T\Delta_n \setminus \left\{ \vec{0}\right\}.
    \end{equation}
    It is \textbf{conditionally negative semidefinite (CNSD)} if the inequality is not strict.
\label{def:conditionally_negative_definite}
\end{definition}
Thus, for our purposes, \cref{lemma:second_order} becomes the following.
\begin{proposition}
    $\entropyname{\vec{Z}}{\alpha}$ is locally strictly concave at $\vec{p} \in \Delta_n^\circ$ if $\nabla^2 \entropy{\vec{Z}}{\alpha}{\vec{p}}$ is CND.
\label{prop:concave_if_cnsd}
\end{proposition}
Henceforth, we use the term `locally strictly concave' without explicitly restricting to the simplex. We can evaluate conditional definiteness using the spectrum of the Hessian projected into $T\Delta_n$:
\begin{equation}
    \vec{P} \, \nabla^2 \entropy{\vec{Z}}{\alpha}{\vec{p}} \, \vec{P} \textnormal{, where } \vec{P} = \vec{I} - \frac{1}{n}\vec{1}\vec{1}^\top.
\end{equation}
$\vec{1}$ is always in the kernel of this matrix. $\nabla^2 \entropy{\vec{Z}}{\alpha}{\vec{p}}$ is CND if all other $n-1$ eigenvalues are negative.

\subsubsection{Proof of \cref{thm:entropy3_convex}} \label{sec:proof}
The Hessian of $\entropyname{\vec{Z}}{3}$ evaluated at $\vec{p}$ is
\begin{equation}
\begin{split}
    \nabla^2 \entropy{\vec{Z}}{3}{\vec{p}}  &= - \left( \vec{Z}\,\textnormal{diag} \left( \vec{Zp} \right) + \textnormal{diag} \left( \vec{Zp} \right)\vec{Z} \right. \\
    &\qquad\quad  + \left. \vec{Z} \,\textnormal{diag} \left( \vec{p} \right) \vec{Z} \right).
\end{split}
\end{equation}
The quadratic form $\vec{x}^\top \nabla^2 \entropy{\vec{Z}}{3}{\vec{p}} \vec{x} $ is linear in $\vec{p}$. Thus, for any $\vec{x}\in T\Delta_n \setminus \left\{ \vec{0} \right\}$, if it is negative for $\vec{p}$ and $\vec{q} \in \Delta_n^\circ$, it must be negative for all convex combinations of $\vec{p}$ and $\vec{q}$.
\begin{equation*}
\begin{split}
    \vec{p},\vec{q} &\in C\left(\entropyname{\vec{Z}}{3}\right) \\
    & \implies \theta \vec{p} + (1-\theta)\vec{q} \in C\left(\entropyname{\vec{Z}}{3}\right) \textnormal{ for all } \theta \in [0,1].
\end{split}
\end{equation*}
This concludes the proof.

For $\alpha > 3$, $\vec{x}^\top \nabla^2 \entropy{\vec{Z}}{\alpha}{\vec{p}} \vec{x}$ is not linear in $\vec{p}$ and the same reasoning does not work. Indeed, we find a counterexample for $\alpha=4$ (\cref{sec:counterexample_alpha4}).

\subsubsection{The region of strict concavity}
There is no guarantee that $C\left(\entropyname{\vec{Z}}{3}\right)$ is non-empty (see \cref{sec:region_of_concavity_nonempty} for sufficient conditions). Given $\vec{Z}$, we can map the subset out by sampling $\Delta_n$ densely and checking each point individually. Alternatively, we can sweep its support function to map out the boundary. For $\vec{x} \in \mathbb{R}^n \setminus \left\{ \vec{0} \right\}$,
\begin{equation*}
     \underset{\vec{p} \in \bar{C}\left(\entropyname{\vec{Z}}{3}\right)}{\textnormal{argmax}} \langle \vec{p},\,\vec{x}\rangle
\end{equation*}
is on the boundary of $\bar{C}\left(\entropyname{\vec{Z}}{3}\right)$. Thus, solving the optimisation problem in many `directions' $\vec{x}$ yields a cloud of boundary points. The problem is convex and can be solved with standard techniques~\cite{Boyd2004} (see our software release).

\subsection{Comparison with optimal transport} \label{sec:methods_ot}

The first experiment tests the effect of the number of categories $n$. We initialise each run by
\begin{enumerate}
    \item embedding each of the $n$ categories uniformly at random into the unit cube $[0,1]^{10}$ equipped with the Euclidean metric to generate structure and
    \item sampling $100$ pairs of compositions from the uniform distribution over $\Delta_n$.
\end{enumerate}
The task is to compare each pair of compositions. We test the following algorithms.
\begin{enumerate}[nosep]
    \item \scriptsize{OT}: Iterate over the 100 pairs and compute the Wasserstein distance~\cite{Vaserten1969,Kantorovich1960} as implemented by the \texttt{emd2} function~\cite{Bonneel2011} of the Python library POT~\cite{Flamary2021}.
    \item \scriptsize{DIV.~(NAIVE)}: Iterate over the 100 pairs and compute structure-aware divergence.
    \item \scriptsize{DIV.~(FAST)}: Compute the divergence of all pairs at once.
\end{enumerate}
{\scriptsize{OT}} works directly with distances and we compute $\vec{Z}$ for the divergences using \cref{eq:z_from_d_tau1}. The Euclidean metric is of negative type, guaranteeing $\vec{Z}\succ0$. For each algorithm, we record the average time per comparison. We report the results of 50 runs each for $n=10$, $100$, and $1000$.

The second experiment is the all-pairs task. We fix $n=50$ and initialise each run by
\begin{enumerate}[nosep]
    \item generating structure in the same way as above and
    \item sampling $m$ compositions from the uniform distribution over $\Delta_n$.
\end{enumerate}

The algorithms do the following.
\begin{enumerate}[nosep]
    \item \scriptsize{OT} iterates over the ${m \choose 2} = m(m-1)/2$ pairs and computes the Wasserstein distance as above.
    \item \scriptsize{DIV.~(NAIVE)} iterates over the pairs and computes the structure-aware information of each pair with equal weights.
    \item \scriptsize{DIV.~(FAST)} is vectorised and computes all the informations in one go.
\end{enumerate}
We record the average time per comparison: the total time divided by $m\choose2$. We report the results of 10 runs each for $m=10$, 31, 100, 316, 1000. We include the code to run both these experiments in our software release.

\subsection{The geography of occupation composition}

\subsubsection{Data}
Census 2021 data for England and Wales is maintained by the Office for National Statistics (ONS) and provided through Nomis. We obtain occupation data from Table TS064, which records the occupation of residents aged 16 years and over~\cite{NomisOccupation}. We pick the Local Authority (LA) level for geography (\textit{local authorities: district / unitary (as of April 2023)} on Nomis). We obtain the geography shapefile from the ONS~\cite{ONSShapefile}.

\subsubsection{Tradable occupations}\label{sec:tradable_occupations}
Some occupations such as teachers, retail workers, and health service staff cater primarily to the local population. The employment in these `non-tradable' occupations scales roughly with population. To characterise regions, we focus on tradable occupations, whose share of the labour force reveals what the region specialises in. Examples include manufacturing and agriculture, as well as services such as legal and finance professions. We follow the methodology of \textcite{Daniotti2025} to identify them.

We start with 104 occupations encoded at the 3-digit level of the Standard Occupational Classification (SOC) 2020. To quantify the tradability of occupation $i$, we measure the correlation between the distributions of $i$'s workers and all workers over LAs. A high Pearson's $r$ indicates that $i$ employs similar shares of the labour force in each LA. We define tradability as $1-r$~\cite{Daniotti2025}. The Spearman's rank correlation between tradability and Krugman's locational Gini coefficient is 0.97, indicating that the latter would rank occupations almost identically (\cref{fig:ew_tradable}).

Picking a threshold of $0.2$, we retain 42 occupations that employ 24.2\% of the workforce (\cref{fig:ew_tradable}). Additionally, we drop \textit{Bed and Breakfast and Guest House Owners and Proprietors} since it has no employment in several LAs. It is the smallest of the tradable occupations and only employs 7,354 people (0.026\% of the workforce). This leaves us with 41 occupations (\cref{tab:tradable_occs}).

\subsubsection{Similarity matrices}\label{sec:occ_similarity_matrices}
\paragraph{Skills similarity.} We define skills similarity with the U.S. Department of Labor's O$^*$NET mapping of occupations to importance ratings between 1 and 5 for each of 35 skills~\cite{ONETSkills} (\cref{sec:occ_similarity_appendix}). We use the National Foundation for Educational Research (NFER) mapping to connect SOC2020 to O$^*$NET occupation codes~\cite{NFERMapping} and aggregate to the 3-digit level. Each of our occupations maps to multiple O$^*$NET occupations. We assign skill vectors in $\mathbb{R}^{35}$ by taking the unweighted mean of the matched O$^*$NET occupations' skill vectors. Of the 394 O$^*$NET occupations in the matching, 20 do not have skills data. This only affects 12 of our occupations, each of which has multiple other O$^*$NET occupations to estimate skills from.

\paragraph{Co-location similarity.} To define co-location similarity, we compute the Revealed Comparative Advantage (RCA) of LA $l$ in occupation $i$ using the Balassa index~\cite{Balassa1965}.
\begin{equation}
    \textnormal{RCA}(i, l) = \frac{N_{il}/N_{\cdot l}}{N_{i\cdot}/N_{\cdot\cdot}},
\end{equation}
where $N_{il}$ is the number of people in $l$ that work in $i$ and $\cdot$ indicates a sum over the index. $\textnormal{RCA}(i, l)$ is the ratio of $i$'s share in the workforce of $l$ to $i$'s share of the overall workforce. We consider occupations as similar if LAs have similar RCA in them. To avoid noise from small population sizes, we restrict the comparison to the 100 most populous LAs. All of the RCAs are positive and we embed each occupation in $\mathbb{R}^{100}$ as the vector of the log RCA of the LAs in it.

In both cases, we transform Euclidean distance into similarities using $d \mapsto \exp(-\tau d)$ and pick $\tau$ such that the median similarity between occupation pairs is $0.1$ (\cref{fig:occs_Z}). Euclidean distance is of negative type, guaranteeing that we can define Bregman geometries for $\alpha=1$ and $2$. For both similarity matrices, all the LAs' distributions are within the region of strict concavity of $\entropyname{\vec{Z}}{3}$.

\subsubsection{Regionalisation}\label{sec:kmeans}

We employ a $k$-means-style clustering algorithm, which is guaranteed to find locally optimal solutions~\cite{Banerjee2004}. Starting from $k$ initial cluster centres, we alternately update cluster assignments and centres until convergence or reaching a maximum number of iterations. We define convergence as an increase in the fraction of information explained of less than $10^{-12}$ between successive iterations and set the maximum number of iterations to $500$. We report the best partition out of $100$ trials with random initialisations. Our clustering is not explicitly geographical: we do not require clusters to be contiguous. See \textcite{Chodrow2017} for heuristics for spatially-constrained regionalisation based on Bregman information.

\subsection{Ecological $\beta$-diversity}

\subsubsection{Data}
The data was collected by \textcite{Caccianiga2006} and we accessed it using the R package adiv~\cite{Pavoine2020}.

\subsubsection{Species similarity matrix}\label{sec:species_similarity_matrix}
We follow \textcite{Ricotta2021b} to define $\vec{Z}$ using canopy height, leaf dry mass content, leaf dry weight, specific leaf area, leaf nitrogen content, and leaf carbon content. We standardise the traits to zero mean and unit standard deviation, compute the Euclidean distance between the species, and linearly rescale it such that the maximum distance is $1$. We map distances to similarities using $d \mapsto 1-d$. This mapping produces $\vec{Z}\succ0$ for this dataset, but it does not guarantee this in general.

\vspace{1cm}

\noindent \textbf{Acknowledgements.} We thank Karel Devriendt for helpful discussions. RS is funded by the Mathematical Institute at the University of Oxford. RL acknowledges support from the EPSRC grants EP/V013068/1, EP/V03474X/1, and EP/Y028872/1. The funders played no role in designing, executing, or communicating the research.

\noindent \textbf{Conflicting interests.} We declare no conflicting interests.

\noindent \textbf{Author contributions.} Both authors contributed to conceptualising the research and writing the article. RS developed the methodology, executed the analysis, and wrote the first draft. RL supervised the research.

\noindent \textbf{Code and data availability.} All the data used in this article is public. Open-source software to implement our methods and replicate our results is available at \url{https://github.com/rohit-sahasrabuddhe/parameter-structure-concavity}.

\noindent \textbf{Generative AI use.} We used OpenAI's ChatGPT (versions 5.1 - 5.6) to aid in (i) brute force searches for counterexamples (\cref{sec:counterexample_alpha4,sec:counterexample_ricotta}), (ii) vectorising the {\scriptsize{DIV.~(FAST)}} algorithm for the all-pairs task (\cref{sec:methods_ot}), and (iii) polishing the exposition. We checked the counterexamples and {\scriptsize{DIV.~(FAST)}} manually; both are public in our software release. We did not include any block of text directly from an LLM output.

\printbibliography

\onecolumn
\appendix
\crefalias{section}{appendix}
\crefalias{subsection}{appendix}
\setcounter{figure}{0}
\setcounter{proposition}{0}
\setcounter{equation}{0}
\setcounter{lemma}{0}
\setcounter{table}{0}

\renewcommand{\thefigure}{S\arabic{figure}}
\renewcommand{\thetable}{S\arabic{table}}
\renewcommand{\theproposition}{S\arabic{proposition}}
\renewcommand{\thelemma}{S\arabic{lemma}}
\renewcommand{\theequation}{S\arabic{equation}}

\begin{center}
    {\Huge Appendices}
\end{center}

\section{Mathematical supplement}

\subsection{$\alpha$-logarithmic entropy is strictly concave}\label{sec:entropy_is_concave}
\begin{proposition}
    $\entropyname{}{\alpha}$ is strictly concave in $\Delta_n^\circ$ for $\alpha>0$.
\end{proposition}
\begin{proof}
    $\entropyname{}{1}$ is Shannon entropy.
    \begin{equation}
        \entropy{}{1}{\vec{p}} = - \sum_{i=1}^n p_i \ln p_i
    \end{equation}
    It is a linear combination of the function $f:p\mapsto -p\ln p$, with second derivative $f''(p)=-1/p < 0$  in $\Delta_n^\circ$. Thus, $\entropyname{}{1}$ is strictly concave.
    
    For $\alpha \neq 1$,
    \begin{equation}
        \entropy{}{\alpha}{\vec{p}} = \frac{1}{1 - \alpha} \left( \sum_{i=1}^n p_i^\alpha -1\right).
    \end{equation}    
    Let us isolate the function $f:p \mapsto p^\alpha$ for $p\in(0,1)$, which is strictly concave for $\alpha\in(0,1)$ and strictly convex for $\alpha>1$. $\entropyname{}{\alpha}$ is a linear combination of instances of $f$ and thus inherits its curvature depending on the sign of $1-\alpha$ (see \textcite[Section 3.2]{Boyd2004}). The curvature is the same as $f$ for $\alpha \in (0,1)$ and flips for $\alpha>1$. In either case, $\entropyname{}{\alpha}$ remains strictly concave.
\end{proof}

\subsection{Coarse-graining is a limiting case of structure-awareness}\label{sec:coarse-graining_limit}

We use coarse-graining to refer to grouping categories together. Consider a coarse-graining into $m<n$ categories, where we will write $i \in a$ or $a \in i$ to indicate that $i$ is coarse-grained to $a$. We represent this as a transformation
\begin{equation}
    \begin{split}
       T:\Delta_n &\longrightarrow \Delta_m \\
        \vec{p} &\longmapsto \tilde{\vec{p}} = \vec{Tp},
    \end{split}
\end{equation}
where $\vec{T} \in \left\{0,1 \right\}^{m \times n}$ such that
\begin{equation}
    T_{ai} = \begin{cases}
        1 & \textnormal{if } i \in a\\
        0 & \textnormal{otherwise.}
    \end{cases}
\end{equation}
We can represent this as structure-awareness using the similarity matrix $\vec{Z}=\vec{T}^\top \vec{T}$, the indicator matrix for whether $i$ and $j$ are in the same coarse-grained category. The ordinariness of $i$ is the probability of its group.
\begin{equation}
    \left(Zp\right)_i = \sum_{j=1}^n Z_{ij} p_j = \sum_{j \in a} p_j = \tilde{p}_a \textnormal{ for } a \owns i
\end{equation}

\begin{lemma}
    For a coarse-graining $T$, $\vec{Z} = \vec{T}^\top \vec{T}$, and any function $f:\mathbb{R}\to\mathbb{R}$,
    \begin{equation}
        \sum_{i=1}^n p_i f\left( \left(Zp\right)_i \right) = \sum_{a=1}^m \tilde{p}_a f(\tilde{p}_a)
    \end{equation}
    for all $\vec{p}\in\Delta_n$.
\end{lemma}
\begin{proof}
    This follows immediately from the definitions.
    \begin{align*}
        \sum_{i=1}^n p_i f\left( \left(Zp\right)_i \right) &= \sum_{a=1}^m f(\tilde{p}_a) \left( \sum_{i\in a}  p_i\right)\\
        &= \sum_{a=1}^m \tilde{p}_a f(\tilde{p}_a)
    \end{align*}
\end{proof}

Setting $f=\sigma_\alpha$, the surprise function of the $\alpha$-logarithmic entropy, gives us
\begin{equation}
    \entropy{\vec{Z}}{\alpha}{\vec{p}} = \entropy{}{\alpha}{\tilde{\vec{p}}}.
\label{eq:entropy_coarsegrained}
\end{equation}

What does this mean for the divergences? $\vec{Z}\succeq 0$ but $\not\succ 0$ since it is not full rank. Moving probability mass between $i,j\in a$ has no effect on $\entropyname{\vec{Z}}{\alpha}$: it is concave but not strictly concave. Thus, the functions $\divergencename{\vec{Z}}{\alpha}$ are not divergences; we will only use the notation to refer to the functional form. From \cref{eq:entropy_coarsegrained} and the construction of Bregman divergences,
\begin{equation}
    \divergence{\vec{Z}}{\alpha}{\vec{p}}{\vec{q}} = \divergence{}{\alpha}{\tilde{\vec{p}}}{\tilde{\vec{q}}}
\end{equation}
for all $\vec{p}\in\Delta_n$ and $\vec{q}\in\Delta_n^\circ$.

We refer to coarse-graining as a limiting case of structure-awareness because we can get arbitrarily close to $\vec{Z}=\vec{T}^\top\vec{T}$ with the construction
\begin{equation}
    \vec{Z}_\epsilon = (1-\epsilon)\vec{Z} + \epsilon \vec{I} \textnormal{ for } \epsilon \in (0,1]
\end{equation}
$\vec{Z}_\epsilon \to \vec{Z}$ as $\epsilon \to 0^+$. It is also a similarity matrix, is positive definite, and satisfies the multiplicative triangle inequality. Thus, coarse-graining is the limit of a family of similarity matrices with block structure that tunes within-block distinctions.

\subsection{The role of $\alpha$} \label{sec:role_of_alpha}

First, we show that increasing $\alpha$ reduces the surprise of every category (except in the trivial case where its ordinariness is $1$).

\begin{proposition} \label{prop:alpha_decreases_surprise}
    Let $\sigma_\alpha$ be the surprise function $p\mapsto \ln_\alpha(1/p)$. For $p\in(0,1)$, the map
    \begin{equation*}
        \alpha \mapsto \sigma_\alpha(p)
    \end{equation*}
    is strictly decreasing on $[0, \infty)$.
\end{proposition}
\begin{proof}
    \begin{equation*}
    \sigma_\alpha(p) = \begin{cases}
        \frac{1}{\alpha - 1} \left( 1 - p^{\alpha-1} \right) & \textnormal{for } \alpha \neq 1 \textnormal{ and}\\
        - \ln p & \textnormal{for } \alpha=1
        \end{cases}
    \end{equation*}
    For simplicity, we will only consider the $\alpha\neq1$ case; the $\alpha=1$ case follows by continuity.
    \begin{align*}
        \frac{d}{d\alpha} \sigma_\alpha(p) &= -\frac{1}{\alpha-1} p^{\alpha-1}\ln p - \frac{1}{(\alpha-1)^2}\left( 1 - p^{\alpha-1} \right)\\
        &= \frac{1}{(\alpha-1)^2} \left( p^{\alpha-1} \left( 1 - \ln p^{\alpha-1} \right) - 1 \right),
    \end{align*}
    where we substituted $\ln p = \frac{1}{\alpha-1} \ln p^{\alpha-1}$. Now, since $p\in(0,1)$, $p^{\alpha-1} \left( 1 - \ln p^{\alpha-1} \right) < 1$, and therefore, $\frac{d}{d\alpha} \sigma_\alpha<0$.
\end{proof}

Now, we show that increasing $\alpha$ affects less ordinary categories more.

\begin{proposition} \label{prop:alpha_tunes_sensitivity}
    For $\alpha>0$, the map
    \begin{equation*}
        p \mapsto \frac{d}{d\alpha} \sigma_\alpha(p)
    \end{equation*}
    is strictly increasing on $(0,1)$.
\end{proposition}
\begin{proof}
    \begin{align*}
        \frac{d}{dp} \left( \frac{d}{d\alpha} \sigma_\alpha(p) \right) = -p^{\alpha-2} \ln p
    \end{align*}
    Since $p\in(0,1)$, we have $p^{\alpha-2} >0$ and $\ln p<0$. Therefore, $\frac{d}{dp} \left( \frac{d}{d\alpha} \sigma_\alpha(p) \right) > 0$.
\end{proof}

\subsection{From metrics to similarity matrices}
\begin{proposition}
    A matrix $\vec{Z}$ is a similarity matrix satisfying the multiplicative triangle inequality (see \cref{conjecture}) if there is a metric $d$ such that
    \begin{equation}
        Z_{ij} = \exp(-d(i, j)).
    \end{equation}
    If $Z_{ij} \in (0,1)$ for all $i\neq j$, this condition is also necessary.
\label{prop:metric_similarity}
\end{proposition}
\begin{proof}
    For sufficiency, we must prove that $\vec{Z}$ is a similarity matrix and that it satisfies the multiplicative triangle inequality. The former is straightforward: the symmetry and diagonal are inherited respectively from the symmetry and identity of indiscernibles of the metric, and $0 \leq \vec{Z} \leq 1$ element-wise is guaranteed by the range of $x\mapsto \exp(-x)$ for $x\geq0$. The triangle inequality of the distances induces the multiplicative version for the similarities.
    \begin{equation}
        Z_{ij} \times Z_{jk} = \exp\left(- (d(i,j) + d(j,k)) \right) \leq \exp\left(- d(i,k) \right) = Z_{ik}.
    \end{equation}
    This proves sufficiency.
    
    The transformation $x\mapsto \exp(-x)$ is invertible and thus, we can also prove necessity by construction. Given $\vec{Z} > 0$ element-wise, let $d$ be a function such that $d(i,j) = -\ln Z_{ij}$.
    \begin{equation}
        d(i,j) + d(j,k) = -\ln (Z_{ij}\times Z_{jk}) \geq -\ln Z_{ik} = d(i,k)
    \end{equation}
    Once again, the symmetry and identity of indiscernibles are inherited directly from the symmetry and diagonal of $\vec{Z}$, and non-negativity from the range of $x \mapsto -\ln x$ for $x\in (0,1]$. Thus, $d$ is a metric.
\end{proof}

\subsection{Sufficient condition for $C\left(\entropyname{\vec{Z}}{3}\right)$ to be non-empty}\label{sec:region_of_concavity_nonempty}
That $C\left(\entropyname{\vec{Z}}{3}\right)$ is convex does not guarantee it is non-empty. For instance, when $\vec{Z}$ is the all-ones matrix, $\entropyname{\vec{Z}}{3}$ is constant and is therefore never strictly concave.

\begin{theorem}
    For a similarity matrix $\vec{Z}$, if $\vec{Z}\succ 0$ and $\vec{Z}^{-1} \vec{1} > 0$ element-wise,
    \begin{equation*}
        \frac{\vec{w}}{\vec{1}^\top \vec{w}} \in C\left( \entropyname{\vec{Z}}{3} \right).
    \end{equation*}
\end{theorem}
\begin{proof}
    Let $\vec{w}= \vec{Z}^{-1} \vec{1}$ and $c = 1  / \vec{1}^\top \vec{w} > 0$, and consider $\vec{p}^* = c\vec{w} \in \Delta_n^\circ$.
    The Hessian for $\alpha=3$ evaluated at $\vec{p^*}$ is as follows.
    \begin{align*}
        \nabla^2 \entropy{\vec{Z}}{3}{\vec{p^*}} &= - \left( \vec{Z}\,\textnormal{diag} \left( \vec{Zp^*} \right) + \textnormal{diag} \left( \vec{Zp^*} \right)\vec{Z} + \vec{Z} \,\textnormal{diag} \left( \vec{p^*} \right) \, \vec{Z} \right) \\
        &= - \left( 2c \vec{Z}  + \vec{Z} \,\textnormal{diag} \left( \vec{p^*} \right) \, \vec{Z} \right)
    \end{align*}
    since $\vec{Zp}^* = c\vec{1}$.

    Since $\vec{Z}\succ 0$ and $c > 0$, $2c \vec{Z} \succ 0$. Further, since $\vec{p^*} > 0$ element-wise, $\textnormal{diag} \left( \vec{p^*} \right) \succ 0$ and therefore $\vec{Z} \,\textnormal{diag} \left( \vec{p^*} \right) \, \vec{Z} \succ 0$. This means that $\nabla^2 \entropy{\vec{Z}}{3}{\vec{p^*}} \prec 0$. Therefore, $\vec{p^*} \in C\left( \entropyname{\vec{Z}}{3} \right)$.
\end{proof}

The distribution $\vec{p^*}$ ties to the theory of magnitude~\cite{Leinster2013}. When $\vec{Z}$ is defined from a metric $d$ as $Z_{ij}=\exp \left(-d(i,j) \right)$, $\vec{w}$ is called the `weighting' and $\vec{1}^\top \vec{w}$ the `magnitude' of the metric space. \textcite{Leinster2013} established several cases when $\vec{w}$ is guaranteed to be positive.
\begin{enumerate}[nosep,]
    \item Every finite subspace of $\mathbb{R}$~\cite[Proposition 2.4.13]{Leinster2013}
    \item Any metric space with at most 3 points~\cite[Proposition 2.4.15]{Leinster2013}
    \item Any `ultrametric' space~\cite[Proposition 2.4.18]{Leinster2013}: one where the triangle inequality is strengthened to
        \begin{equation*}
            \max \left\{d(i,j), d(j,k) \right\} \geq d(i,k).
        \end{equation*}
\end{enumerate}
Making $\tau$ sufficiently large guarantees not only that $\vec{Z}(\tau)\succ0$ but also that $\vec{w} > 0$~\cite[Proposition 2.2.6 (iii)]{Leinster2013}. The same lower bound of the minimum distance being larger than $\ln (n-1)$ applies~\cite[Proposition 2.4.17]{Leinster2013}.

\subsection{$C\left(\entropyname{\vec{Z}}{4}\right)$ is not necessarily convex}\label{sec:counterexample_alpha4}
\begin{proposition}
    $C\left(\entropyname{\vec{Z}}{4}\right)$ is not necessarily convex.
\end{proposition}
\begin{proof}
    We will prove this by counterexample. Let
    \begin{equation*}
        \vec{Z} = \begin{bmatrix}
            1 & 0.05 & 0.04 & 0.41 \\
            0.05 & 1 & 0.62 & 0.02 \\
            0.04 & 0.62 & 1 & 0.57 \\
            0.41 & 0.02 & 0.57 & 1
        \end{bmatrix}
    \end{equation*}
     and consider the distributions
     \begin{align*}
         \vec{p} &= \begin{bmatrix} 0.02 & 0.45 & 0.41 & 0.12\end{bmatrix}^\top \textnormal{ and}\\
         \vec{q} &= \begin{bmatrix} 0.23 & 0.71 & 0.03 & 0.03\end{bmatrix}^\top \textnormal{ such that}\\
         \frac{\vec{p} + \vec{q}}{2} &= \begin{bmatrix} 0.125 & 0.58 & 0.22 & 0.075 \end{bmatrix}^\top.
     \end{align*}

     The eigenvalues of $\vec{P} \nabla^2 \entropyname{\vec{Z}}{4} \vec{P}$ for $\vec{P} = \vec{I}_4 - \vec{1}\vec{1}^\top/4$ are as follows (rounded to the $4^\textnormal{th}$ decimal place):
     \begin{itemize}[nosep]
         \item At $\vec{p}$: -1.3340, -0.6685, -0.0028, 0.0000
         \item At $\vec{q}$: -1.5095, -0.2675, -0.0013, 0.0000
         \item At $\left( \vec{p} + \vec{q} \right)/2$: -1.3728, -0.4417, 0, \textbf{0.0007}
     \end{itemize}
     Thus, $\nabla^2 \entropyname{\vec{Z}}{4}$ is CND at $\vec{p}$ and $\vec{q}$ but not at $\left( \vec{p} + \vec{q} \right)/2$. This concludes the proof.
\end{proof}
We note that the $\vec{Z}$ above is positive definite, which therefore does not guarantee convexity.

\section{Disproving Theorem 2 of \textcite{Ricotta2006}} \label{sec:counterexample_ricotta}

While proposing $\entropyname{\vec{Z}}{\alpha}$, \textcite{Ricotta2006} used distances rather than similarities to encode structure, defining a distance matrix $\vec{D}$ normalised such that $0\leq \vec{D} \leq 1$ element-wise. Setting $\vec{Z} = 1-\vec{D}$ in our formulation recovers theirs. We note that our definition is more general since we do not require $\vec{Z}$ to satisfy a triangle inequality-like property. \textcite[Theorem 2]{Ricotta2006} claimed that $\entropyname{\vec{Z}}{\alpha}$ is concave for $\alpha \geq 2$ if $\vec{Z}=1-\vec{D}\succeq0$.

\subsection{Error in proof}
In their proof, the authors define a diagonal matrix $\vec{R}$ with
\begin{equation*}
    R_{ii} = \left(\lambda (\vec{Zp})_i + (1-\lambda) (\vec{Zq})_i \right)^{\alpha-2}
\end{equation*}
for $\lambda \in (0,1)$ and $\vec{p},\vec{q}\in\Delta_n$. They then write that $\vec{Z}\succeq0$ (by assumption) and $\vec{R}\succeq 0$ (by construction) implies $\vec{RZ} \succeq 0$. This is not true in general and holds only when $\vec{R}$ and $\vec{Z}$ commute.

\subsection{A counterexample}
The error has no easy fix and we find the following counterexample. Consider the following normalised distance matrix between 3 elements.
\begin{equation}
    \vec{D} = 
    \begin{bmatrix}
    0 & 0.72 & 1\\
    0.72 & 0 & 0.3\\
     1 & 0.3 & 0
    \end{bmatrix}
\end{equation}

The corresponding similarity matrix is
\begin{equation}
    \vec{Z} = 1 - \vec{D} =
    \begin{bmatrix}
    1 & 0.28 & 0\\
    0.28 & 1 & 0.7\\
     0 & 0.7 & 1
    \end{bmatrix}.
\end{equation}

$\vec{Z}$ is real, symmetric, and its eigenvalues are $0.25$, $1$, and $1.75$: it is positive definite. Consider the distributions
\begin{align*}
    \vec{p} = \begin{bmatrix}0.9818 & 0.0022 & 0.0160 \end{bmatrix}^\top\\
    \vec{q} = \begin{bmatrix}0.9782 & 0.0178 & 0.0040 \end{bmatrix}^\top
\end{align*}
such that
\begin{equation*}
    \frac{\vec{p}+\vec{q}}{2} = \begin{bmatrix}0.98 & 0.01 & 0.01\end{bmatrix}^\top.
\end{equation*}

For $\alpha=6$, we have:
\begin{align*}
    \entropy{\vec{Z}}{6}{\vec{p}} &= 0.02030653785355432 \textnormal{ and}\\
    \entropy{\vec{Z}}{6}{\vec{q}} &= 0.02025752388043114,
\end{align*}
such that
\begin{equation*}
    \frac{1}{2}\left(\entropy{\vec{Z}}{6}{\vec{p}} + \entropy{\vec{Z}}{6}{\vec{q}} \right) = 0.02028203086699273 \approx 0.0202820.
\end{equation*}
On the other hand,
\begin{equation*}
  \entropy{\vec{Z}}{6}{\frac{\vec{p}+\vec{q}}{2}} = 0.020281839269597257 \approx 0.0202818.
\end{equation*}

Thus,
\begin{equation*}
    \entropy{\vec{Z}}{6}{\frac{\vec{p}+\vec{q}}{2}} < \frac{1}{2}\left(\entropy{\vec{Z}}{6}{\vec{p}} + \entropy{\vec{Z}}{6}{\vec{q}} \right),
\end{equation*}
which shows that $\entropyname{\vec{Z}}{6}$ is not concave. The difference is at the $7^\textnormal{th}$ decimal place, ruling out machine precision as the cause. See \url{https://github.com/rohit-sahasrabuddhe/parameter-structure-concavity} for an implementation.

Rounded to the third decimal place, the eigenvalues of $\vec{P} \, \nabla^2 \entropy{\vec{Z}}{6}{\frac{\vec{p}+\vec{q}}{2}} \, \vec{P}$ are $-3.059$, 0, and $0.004$. Thus, $\nabla^2 \entropy{\vec{Z}}{6}{\frac{\vec{p}+\vec{q}}{2}}$ is not CNSD. The eigenvector of the positive eigenvalue reveals the direction of positive curvature that we used to construct the counterexample.
\FloatBarrier

\section{Comparison with Optimal Transport}
\begin{figure}[h!]
    \centering
    \begin{subfigure}[t]{.9\textwidth}
        \centering
        \includegraphics{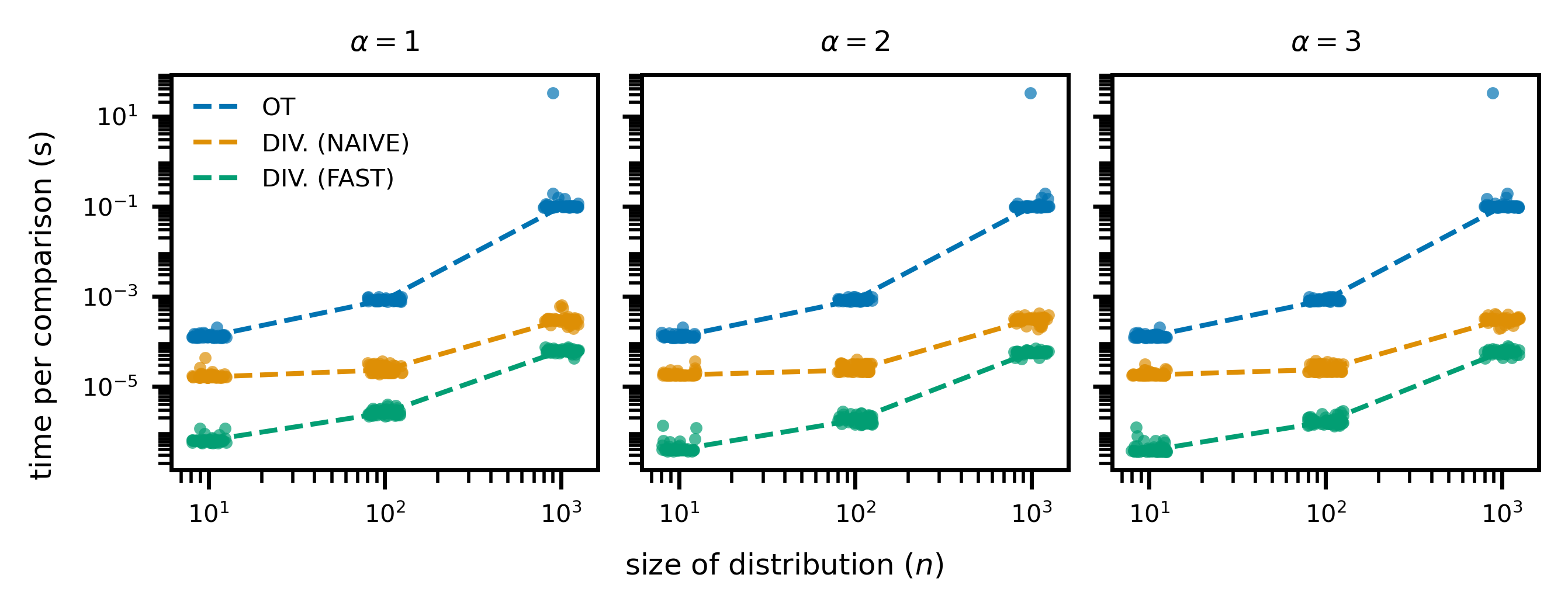}
        \caption{}
    \end{subfigure}
    ~
    \begin{subfigure}[t]{.9\textwidth}
        \centering
        \includegraphics{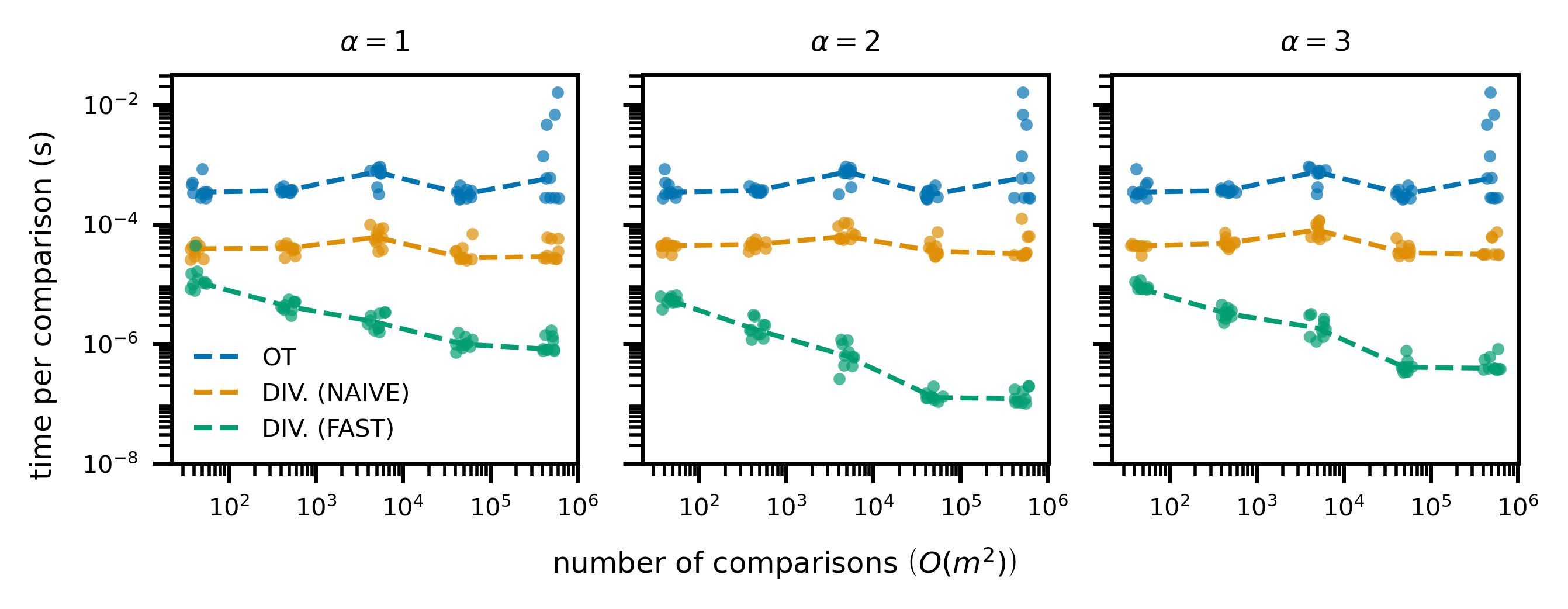}
        \caption{}
    \end{subfigure}   
    \caption{\textbf{Runtime vs optimal transport.} \textbf{(a)} The time per comparison as a function of the number of categories $n$. We conduct 50 runs for each $n$. \textbf{(b)} The time per comparison as a function of the number of comparisons for the all-pairs task. We conduct 10 runs for each $m$. The colour denotes the algorithm. The markers correspond to individual runs and the dashed lines connect the median values. Note the log scale on both axes.}
  \label{fig:ot_appendix}
\end{figure}

\FloatBarrier
\section{The geography of occupation composition}

\subsection{Occupation similarity matrices}\label{sec:occ_similarity_appendix}

The skills in the ONET data are Reading Comprehension, Active Listening, Writing, Speaking, Mathematics, Science, Critical Thinking, Active Learning, Learning Strategies, Monitoring, Social Perceptiveness, Coordination, Persuasion, Negotiation, Instructing, Service Orientation, Complex Problem Solving, Operations Analysis, Technology Design, Equipment Selection, Installation, Programming, Operations Monitoring, Operation and Control, Equipment Maintenance, Troubleshooting, Repairing, Quality Control Analysis, Judgment and Decision Making, Systems Analysis, Systems Evaluation, Time Management, Management of Financial Resources, Management of Material Resources, and Management of Personnel Resources.

Co-located occupations tend to belong to similar sectors, but may require very different skills. For instance, the occupations most similar to \textit{Agricultural and Related Trades} by co-location are \textit{Animal Care and Control Services}, \textit{Veterinary nurses}, and \textit{Managers and Proprietors in Agriculture Related Services}. However, the latter is ranked $20^\textnormal{th}$ out of 40 by skills similarity. The most skills similar occupation is \textit{Elementary Storage Occupations}, which is ranked $26^\textnormal{th}$ by co-location. We will also compare results from modelling no similarity by setting $\vec{Z}=\vec{I}$.

\subsection{Diversity and representativeness of LAs}\label{sec:diversity_LA_appendix}

How does the choice of $\vec{Z}$ and $\alpha$ alter the comparison of LAs by diversity $\left(\entropyname{\vec{Z}}{\alpha}\right)$ and representativeness $\left(\divergencename{\vec{Z}}{\alpha} \textnormal{ from the mean} \right)$?

\paragraph{Diversity.} Setting $\alpha=2$ to align with the notion of coherence~\cite{Daniotti2025}, Kendall's $\tau$ between diversity rankings (\cref{fig:ew_diversity_rankings_Z}) is 0.63 (no similarity vs skills), 0.68 (no similarity vs co-location), and 0.54 (skills vs co-location). Next, we fix $\vec{Z}$ to be skills similarity and vary $\alpha$ (\cref{fig:ew_diversity_rankings_alpha}). Kendall's $\tau$ values are much higher: 0.94 (1 vs 2), 0.90 (1 vs 3), and 0.96 (2 vs 3).

\paragraph{Representativeness.} Fixing $\alpha=2$ (\cref{fig:ew_atypicality_rankings_Z}), the Kendall's $\tau$ values are 0.79 (no similarity vs skills), 0.80 (no similarity vs co-location), and 0.90 (skills vs co-location). Fixing skills similarity (\cref{fig:ew_atypicality_rankings_alpha}), the correlations are 0.94 (1 vs 2), 0.89 (1 vs 3), and 0.95 (2 vs 3). 

\paragraph{}
These results suggest that while the choice of $\vec{Z}$ has a strong effect, varying the relative importance of occupations with rare skills using $\alpha$ does not substantially change the rankings. The latter may be because the range of 1 to 3 for $\alpha$ is not wide enough to capture enough variation for this dataset. To check this, we compute the diversity for $\alpha=10$ and find Kendall's $\tau=0.83$ between it and $\alpha=2$ with skills similarity (\cref{fig:ew_entropy_S_2_10}). That the correlation is much lower than before supports our suspicion.

\begin{figure}[h!]
    \centering
    \begin{subfigure}[t]{\textwidth}
        \centering
        \includegraphics{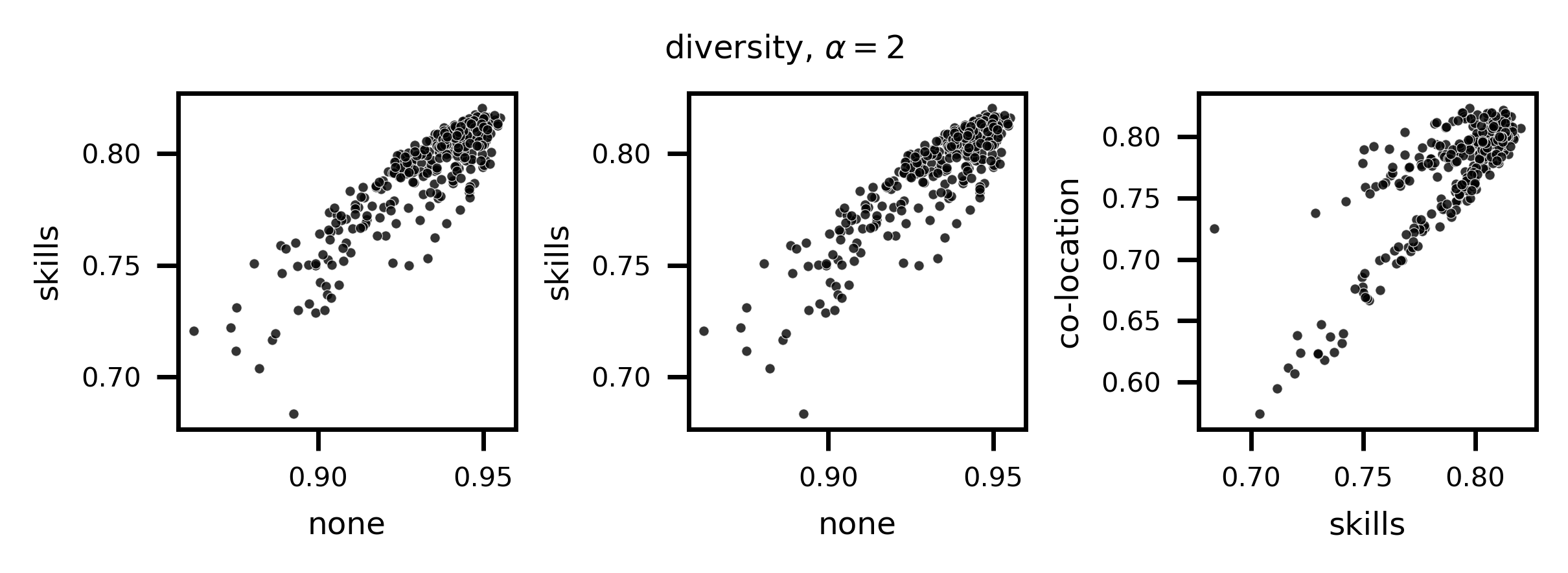}
        \caption{}
    \end{subfigure}\\
    \begin{subfigure}[t]{\textwidth}
        \centering
        \includegraphics{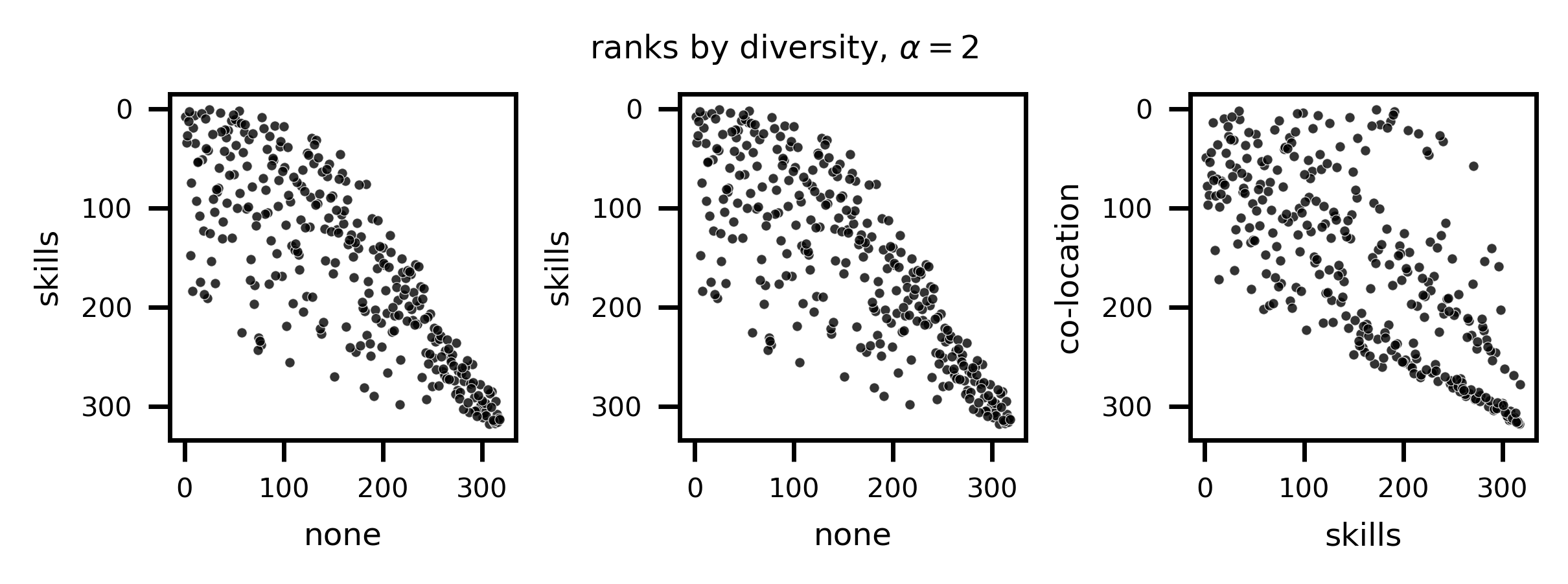}
        \caption{}
    \end{subfigure}
    \caption{\textbf{Effect of $\vec{Z}$ on diversity of LAs.} Each subplot compares the \textbf{(a)} value of diversity, measured as $\entropyname{\vec{Z}}{2}$ and \textbf{(b)} the rankings by diversity for each pair of similarity matrices.}
    \label{fig:ew_diversity_rankings_Z}
\end{figure}

\begin{figure}[h!]
    \centering
    \begin{subfigure}[t]{\textwidth}
        \centering
        \includegraphics{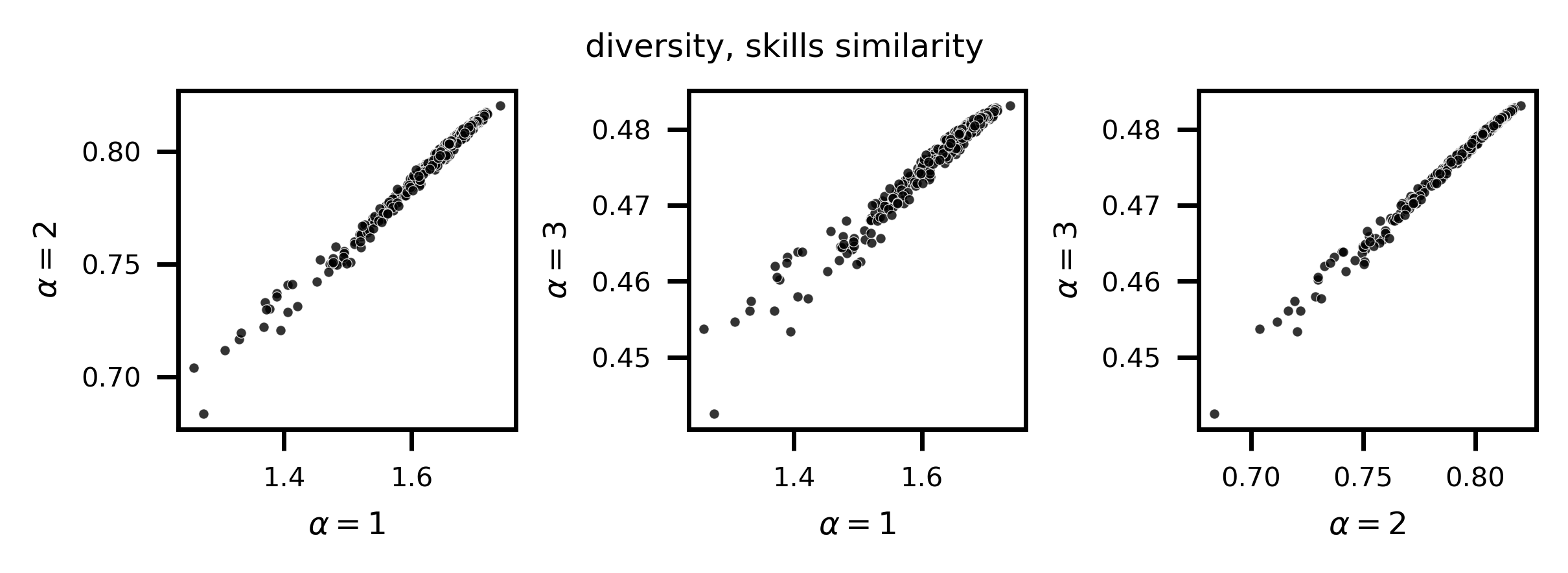}
        \caption{}
    \end{subfigure}\\
    \begin{subfigure}[t]{\textwidth}
        \centering
        \includegraphics{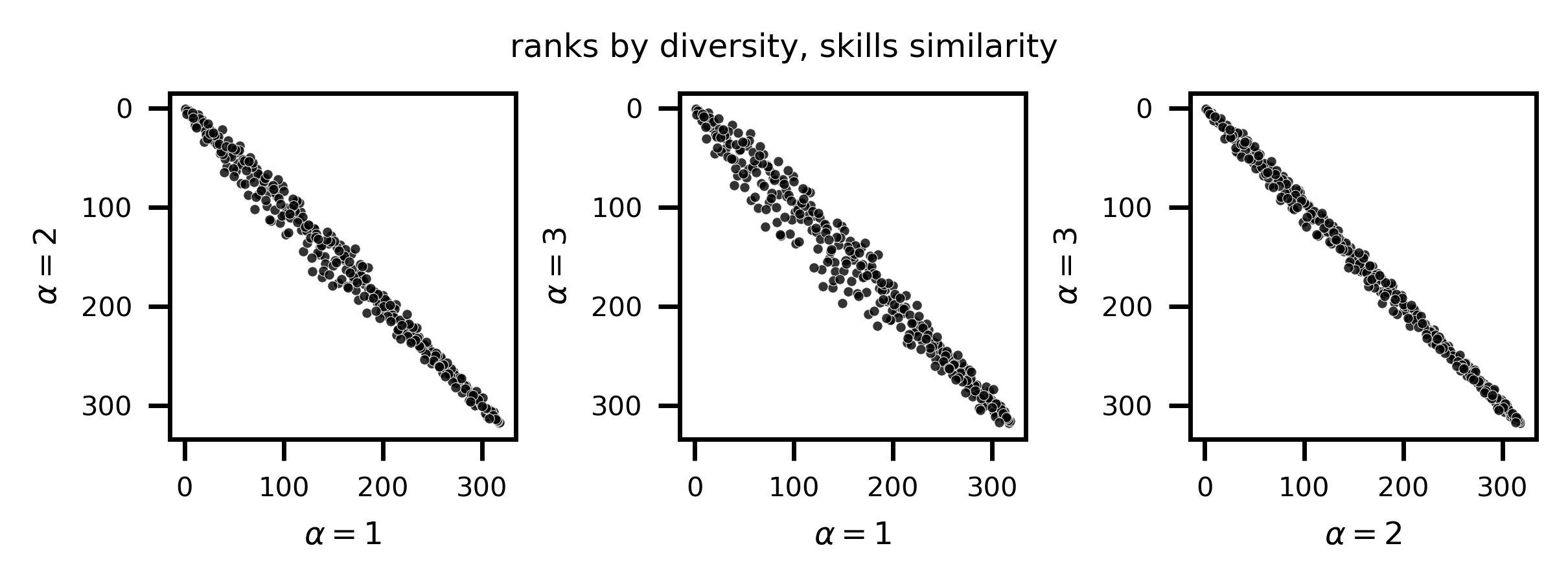}
        \caption{}
    \end{subfigure}
    \caption{\textbf{Effect of $\alpha$ on diversity of LAs.} Each subplot compares the \textbf{(a)} value of diversity, measured as $\entropyname{\vec{Z}}{\alpha}$ and \textbf{(b)} the rankings by diversity for each pair of $\alpha$. We use the $\vec{Z}$ capturing skills similarity.}
    \label{fig:ew_diversity_rankings_alpha}
\end{figure}

\begin{figure}[h!]
    \centering
    \begin{subfigure}[t]{\textwidth}
        \centering
        \includegraphics{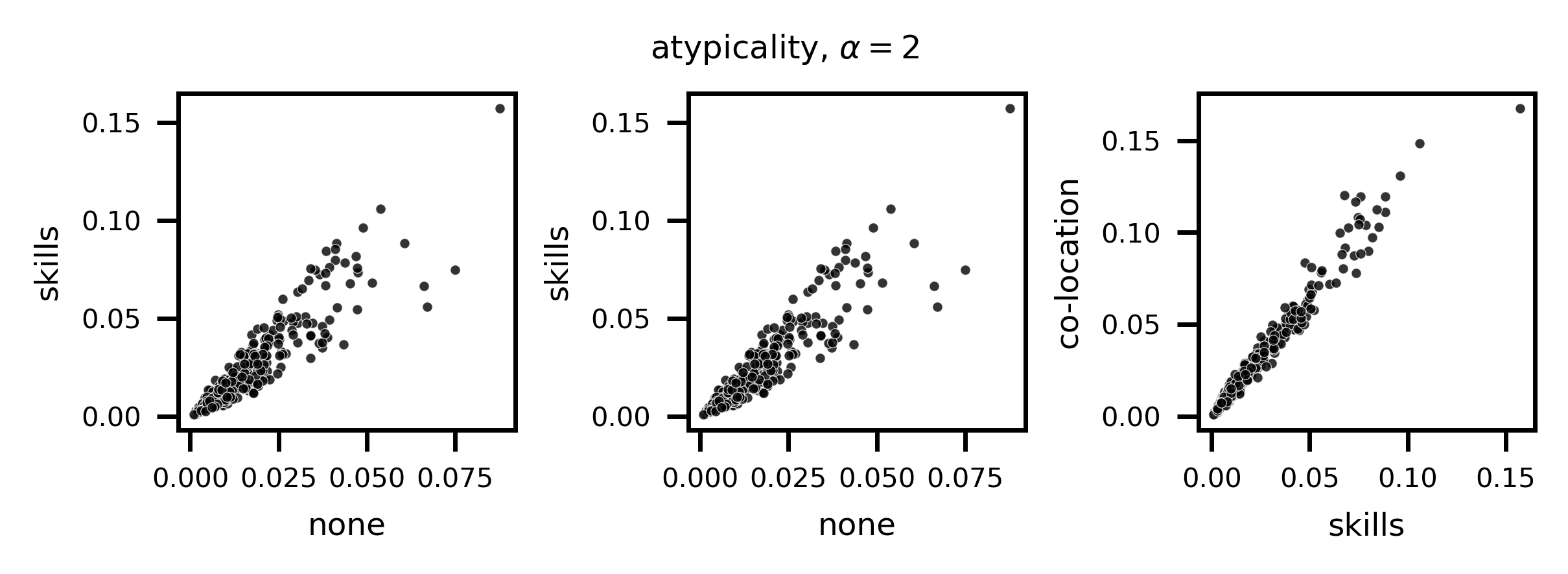}
        \caption{}
    \end{subfigure}\\
    \begin{subfigure}[t]{\textwidth}
        \centering
        \includegraphics{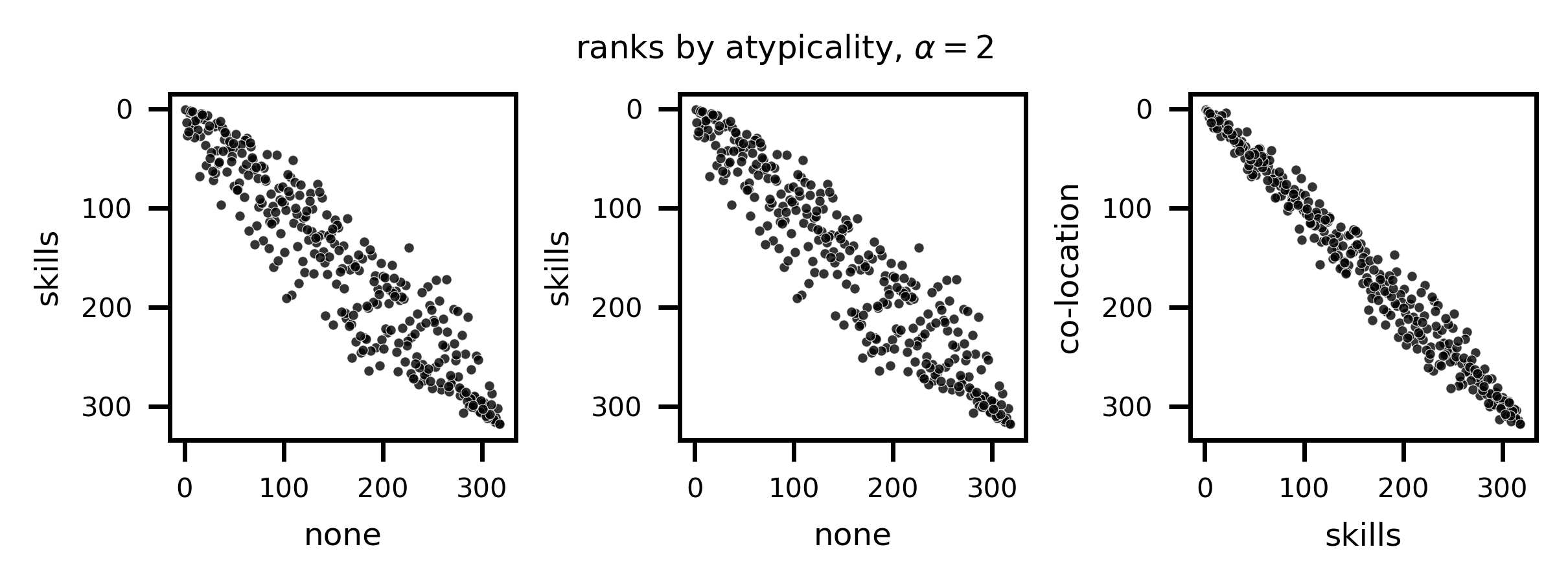}
        \caption{}
    \end{subfigure}
    \caption{\textbf{Effect of $\vec{Z}$ on representativeness of LAs.} Each subplot compares the \textbf{(a)} value of atypicality, measured as divergence $\divergencename{\vec{Z}}{2}$ from the mean and \textbf{(b)} the rankings by atypicality for each pair of similarity matrices.}
    \label{fig:ew_atypicality_rankings_Z}
\end{figure}

\begin{figure}[h!]
    \centering
    \begin{subfigure}[t]{\textwidth}
        \centering
        \includegraphics{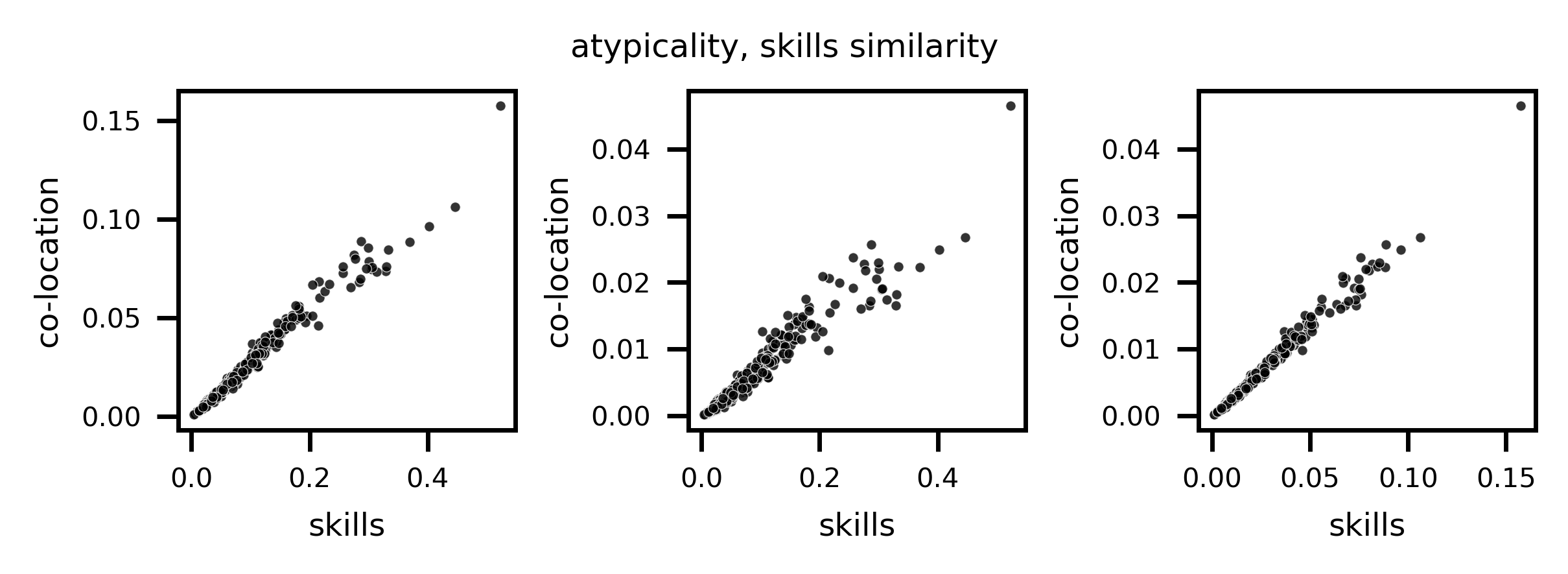}
        \caption{}
    \end{subfigure}\\
    \begin{subfigure}[t]{\textwidth}
        \centering
        \includegraphics{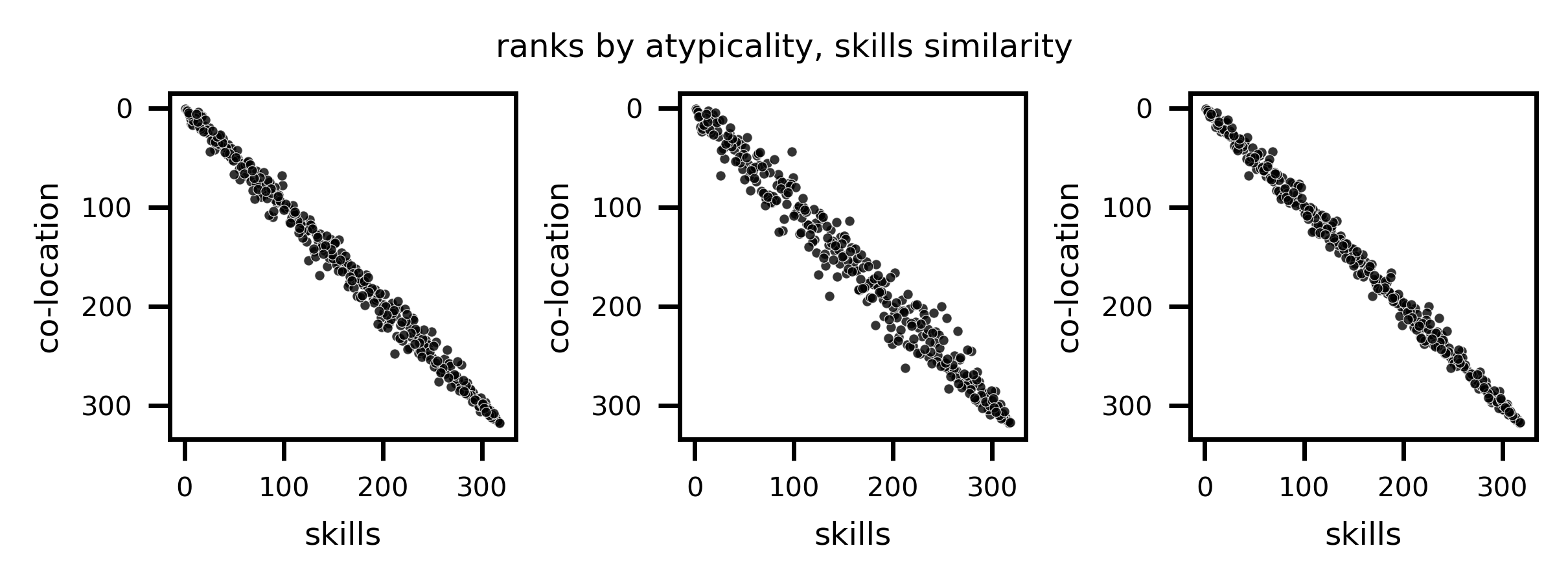}
        \caption{}
    \end{subfigure}
    \caption{\textbf{Effect of $\alpha$ on representativeness of LAs.} Each subplot compares the \textbf{(a)} value of atypicality, measured as divergence $\divergencename{\vec{Z}}{\alpha}$ from the mean and \textbf{(b)} the rankings by atypicality for each pair of $\alpha$. We use the $\vec{Z}$ capturing skills similarity. }
    \label{fig:ew_atypicality_rankings_alpha}
\end{figure}

\begin{figure}[h!]
    \centering
    \includegraphics{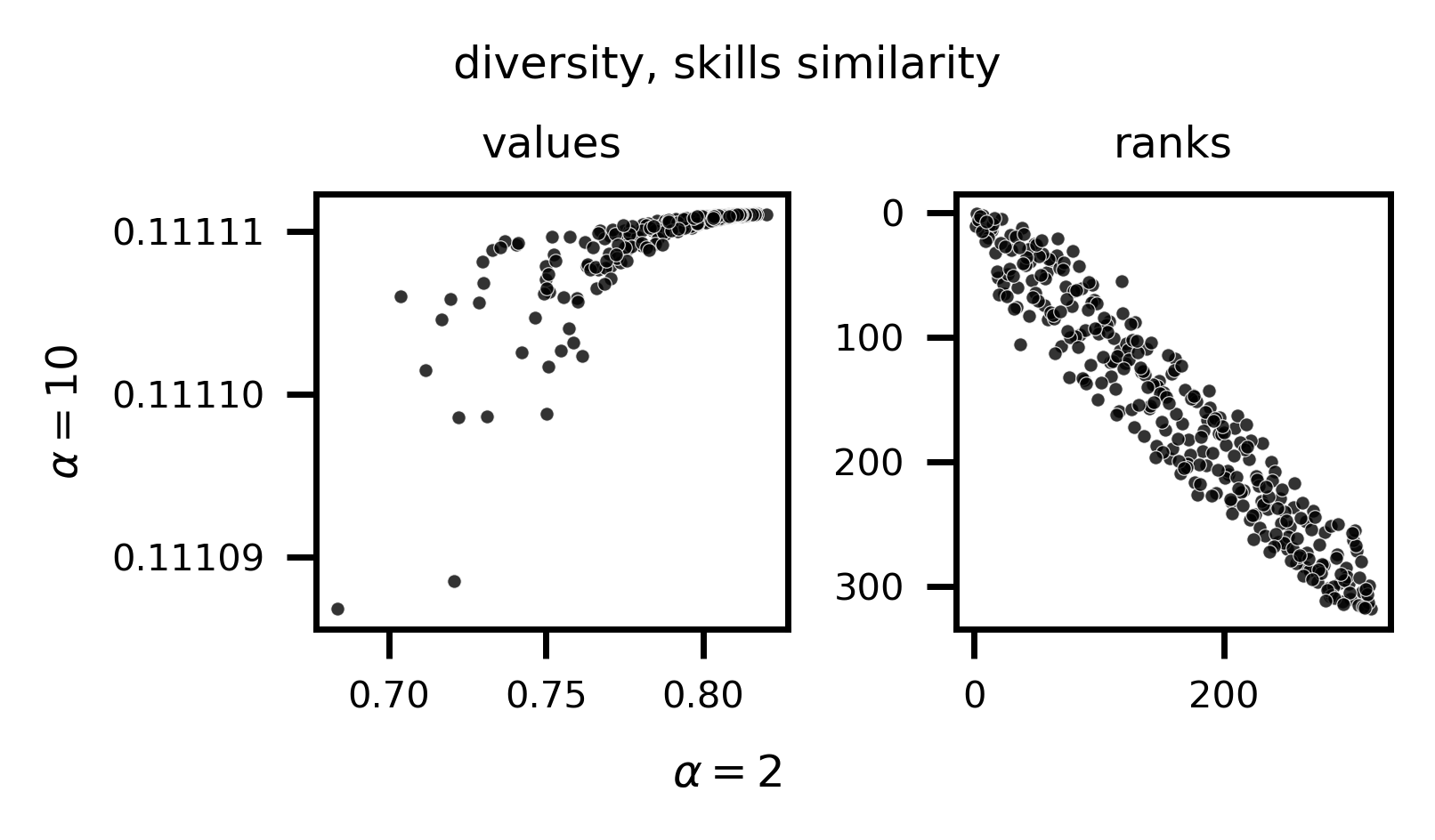}
    \caption{\textbf{Diversity of LAs for $\alpha=2$ and $10$.} Comparing the values of diversity (left), measured as $\entropyname{\vec{Z}}{\alpha}$ and  the rankings (right) for $\alpha=2$ and $10$. We use the $\vec{Z}$ capturing skills similarity.}
    \label{fig:ew_entropy_S_2_10}
\end{figure}

\FloatBarrier

\subsection{Supplementary figures}
\begin{figure}[h!]
    \centering
    \begin{subfigure}[t]{\textwidth}
        \centering
        \includegraphics{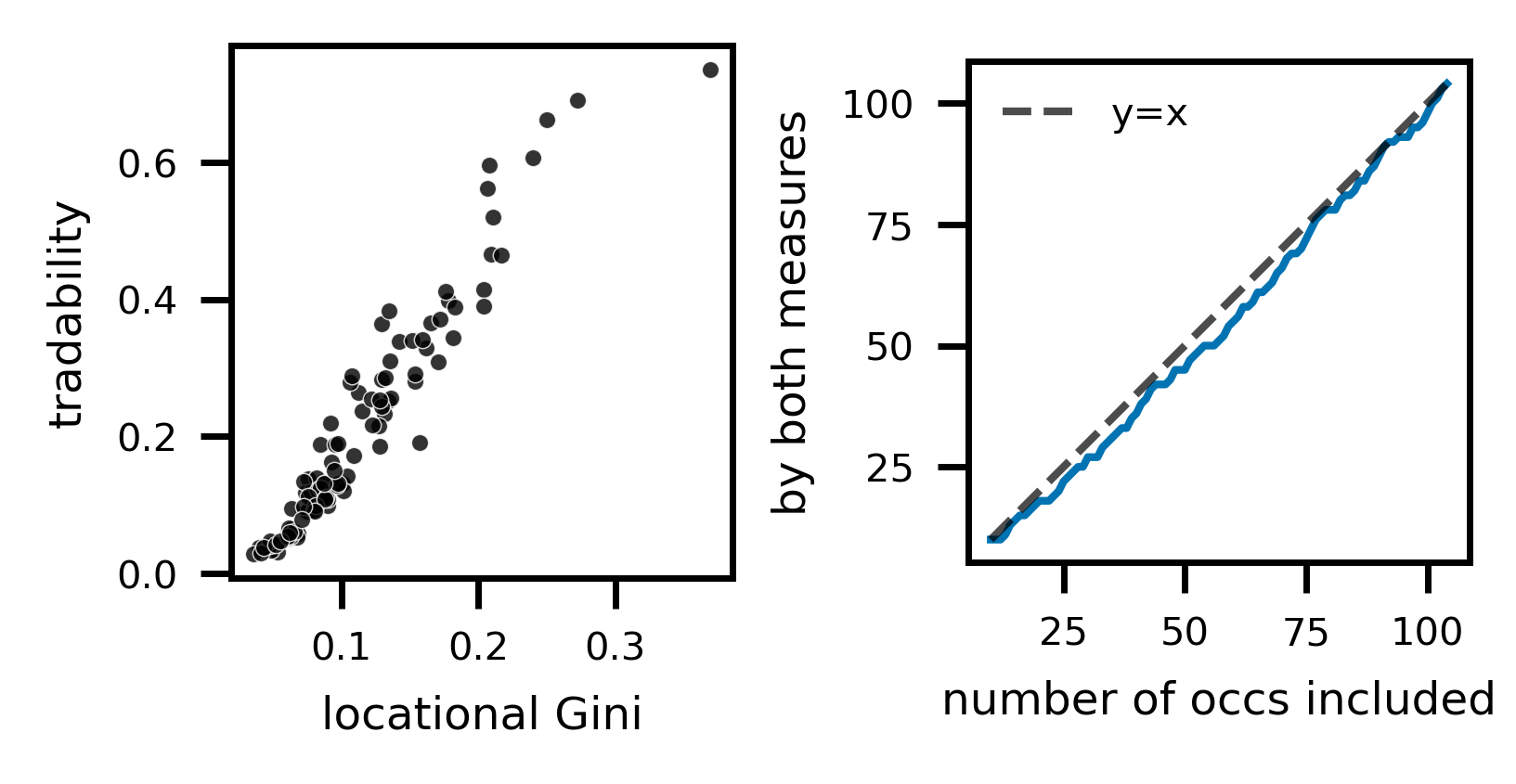}
        \caption{}
    \end{subfigure}\\
    \begin{subfigure}[t]{\textwidth}
        \centering
        \includegraphics{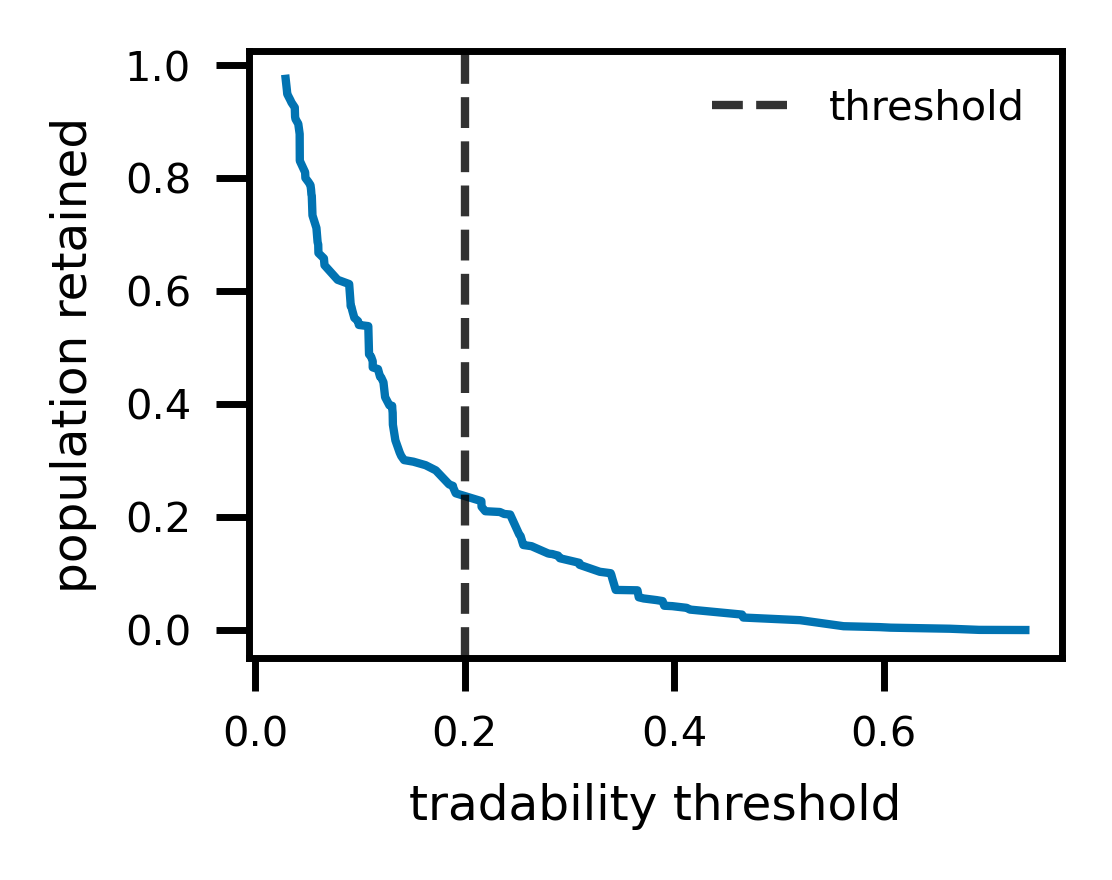}
        \caption{}
    \end{subfigure}
    \caption{\textbf{Identifying tradable occupations.} \textbf{(a)} Tradability against locational Gini index of each occupation (left). Ranking by these indices, we plot the number of occupations that would have been included by both measures against the number of occupations included (right). \textbf{(b)} The fraction of workforce retained as a function of the tradability threshold. We pick a threshold of 0.2, which includes occupations employing 24\% of the workforce. }
    \label{fig:ew_tradable}
\end{figure}

\begin{figure}[h!]
    \centering
    \includegraphics{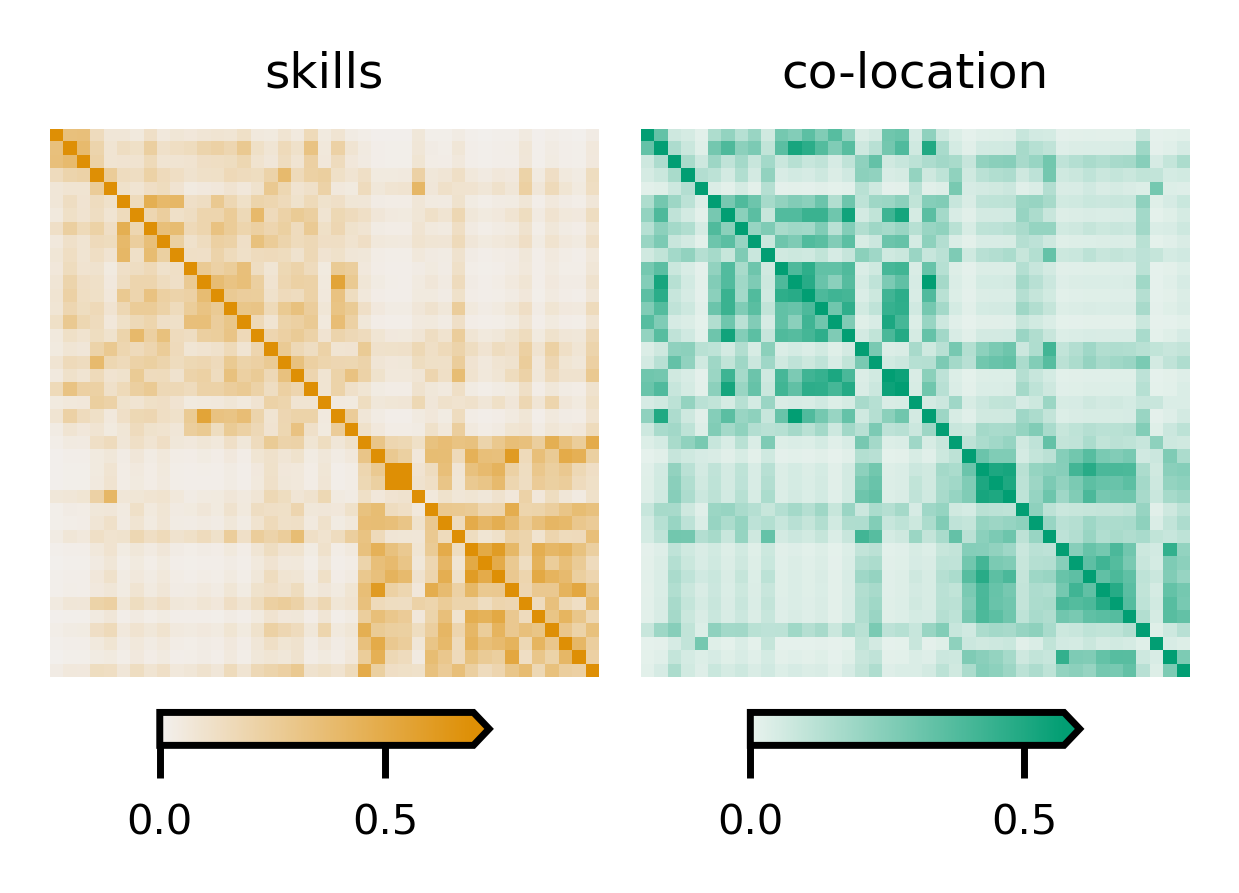}
    \caption{\textbf{Occupation similarity matrices.} $\vec{Z}$ for the 41 tradable occupations modelling skill (left) and co-location (right) similarity.}
    \label{fig:occs_Z}
\end{figure}

\begin{figure}[h!]
    \centering
    \includegraphics[width=.6\textwidth]{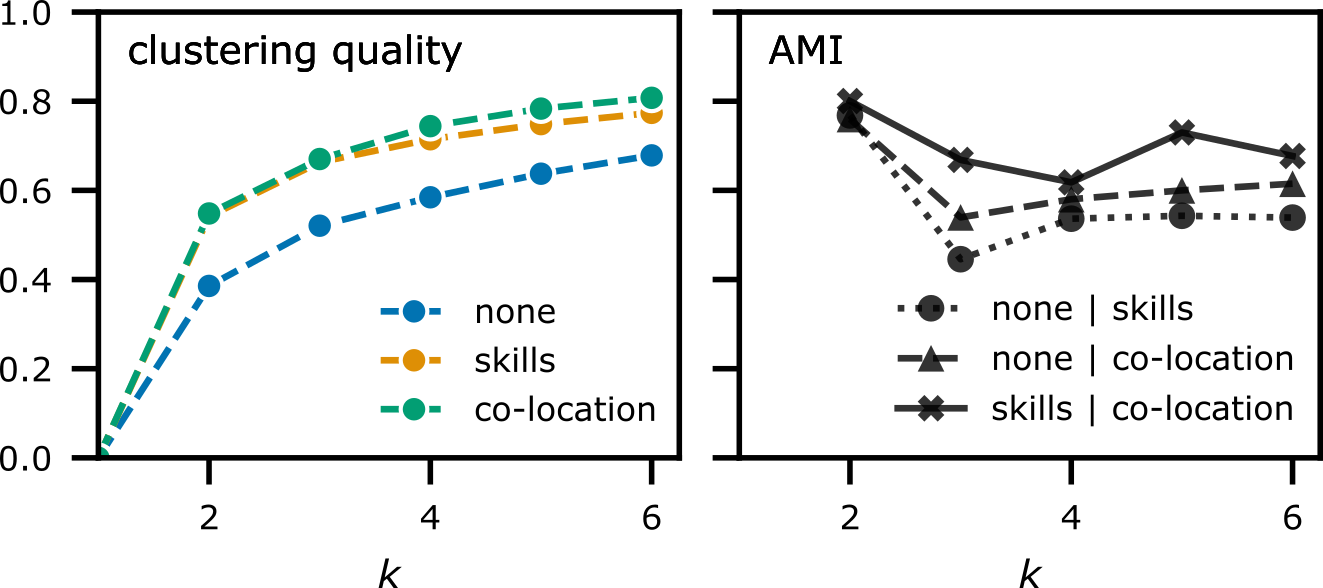}
    \caption{\textbf{Statistics for regionalisation.} (Left) Clustering quality as a function of the number of clusters $k$ for occupation similarity matrices modelling no similarity (blue), skills similarity (orange), and co-location similarity (green). The clustering quality of a partition is the fraction of total Bregman information explained by it. (Right) The Adjusted Mutual Information (AMI) between the optimal partitions for each $k$ for each pair of similarity matrices.}
    \label{fig:ew_clustering_stats}
\end{figure}

\begin{figure}
    \centering
    \includegraphics{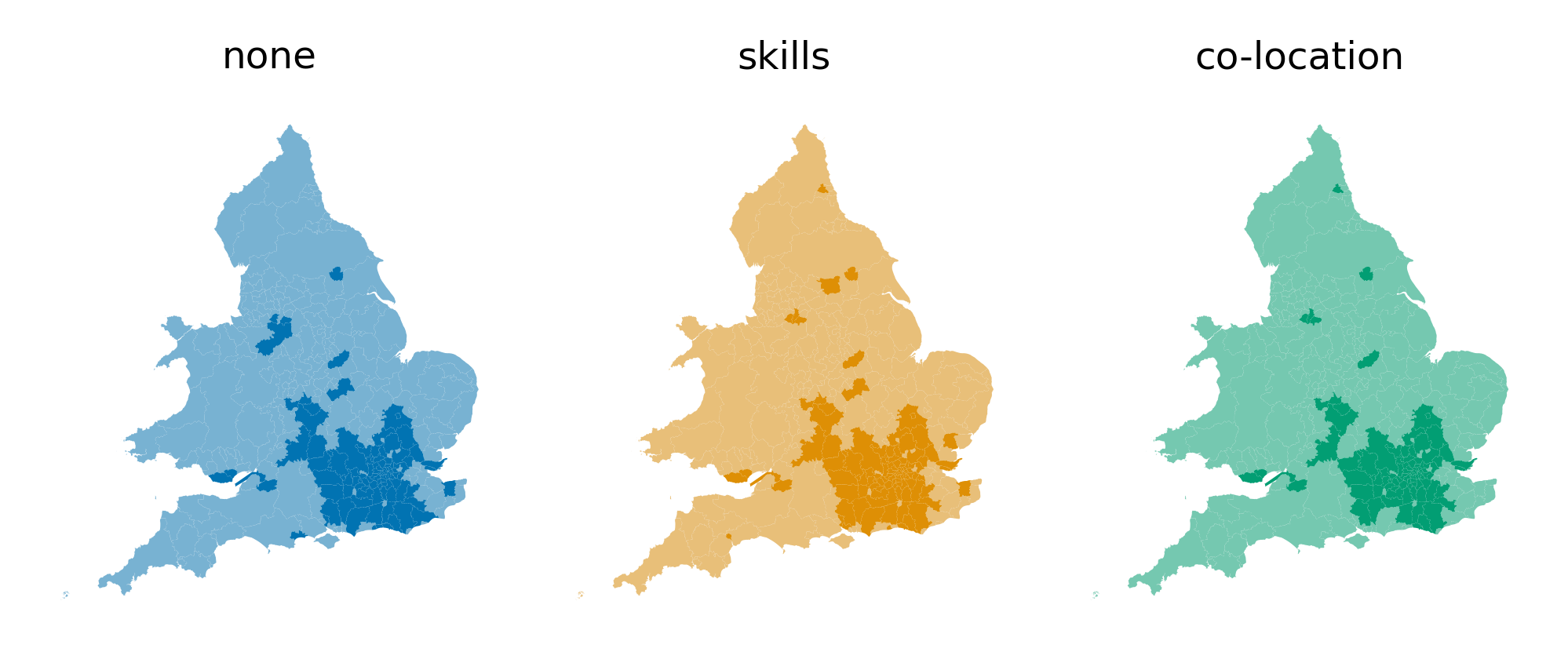}
    \caption{\textbf{Two-way regionalisations by occupation composition.} The geography of the optimal partitions into two regions using no occupation similarity (blue), skills similarity (orange), and co-location similarity (green). We distinguish the regions by intensity of colour. }
    \label{fig:ew_maps_2}
\end{figure}

\FloatBarrier

\subsection{Supplementary tables}
\begin{table}
    \centering
    \begin{tabular}{llr}
    \toprule
    Code & Occupation & Pop. (\%) \\
    \midrule
    111 & Chief Executives and Senior Officials & 1.9 \\
    113 & Functional Managers and Directors & 14.6 \\
    114 & Directors in Logistics, Warehousing and Transport & 0.2 \\
    116 & Senior Officers in Protective Services & 0.7 \\
    121 & Managers and Proprietors in Agriculture Related Services & 0.8 \\
    211 & Natural and Social Science Profs & 1.5 \\
    214 & Web and Multimedia Design Profs & 1.4 \\
    215 & Conservation and Environment Profs & 0.6 \\
    216 & Research and Development (R\&D) and Other Research Profs & 1.3 \\
    224 & Veterinarians & 0.3 \\
    241 & Legal Profs & 3.4 \\
    242 & Finance Profs & 5.7 \\
    243 & Business, Research and Administrative Profs & 5.1 \\
    247 & Librarians and Related Profs & 0.3 \\
    249 & Media Profs & 2.2 \\
    312 & CAD, Drawing and Architectural Technicians & 0.9 \\
    324 & Veterinary nurses & 0.3 \\
    331 & Protective Service Occupations & 5.2 \\
    341 & Artistic, Literary and Media Occupations & 5.1 \\
    342 & Design Occupations & 1.0 \\
    351 & Transport Associate Profs & 0.4 \\
    353 & Finance Associate Profs & 2.8 \\
    411 & Administrative Occupations: Government and Related Orgs & 4.1 \\
    511 & Agricultural and Related Trades & 4.4 \\
    521 & Metal Forming, Welding and Related Trades & 1.2 \\
    522 & Metal Machining, Fitting and Instrument Making Trades & 3.8 \\
    523 & Vehicle Trades & 3.3 \\
    525 & Skilled Metal, Electrical and Electronic Trades Supervisors & 0.6 \\
    541 & Textiles and Garments Trades & 0.6 \\
    542 & Printing Trades & 0.4 \\
    612 & Animal Care and Control Services & 1.2 \\
    811 & Process Operatives & 3.7 \\
    812 & Metal Working Machine Operatives & 0.7 \\
    813 & Plant and Machine Operatives & 1.3 \\
    814 & Assemblers and Routine Operatives & 2.8 \\
    816 & Production, Factory and Assembly Supervisors & 0.9 \\
    822 & Mobile Machine Drivers and Operatives & 1.9 \\
    823 & Other Drivers and Transport Operatives & 1.1 \\
    911 & Elementary Agricultural Occupations & 0.8 \\
    913 & Elementary Process Plant Occupations & 3.4 \\
    925 & Elementary Storage Occupations & 8.2 \\
    \bottomrule
    \end{tabular}
    \caption{\textbf{Tradable occupations in England and Wales.} Profs: Professionals. Orgs: Organisations.}
    \label{tab:tradable_occs}
\end{table}

\FloatBarrier

\section{Ecological $\beta$-diversity}

\subsection{Supplementary figures}

\begin{figure}[h!]
    \centering
    \includegraphics{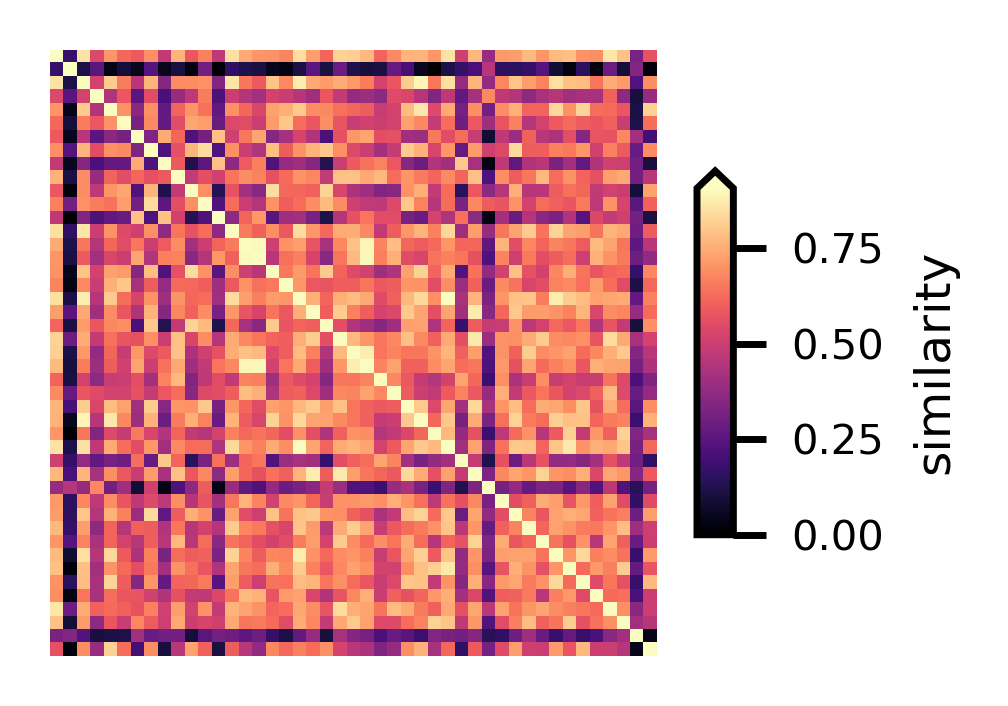}
    \caption{\textbf{Functional similarity matrix for vegetation species in the Rutor glacier.} The species similarity matrix defined using functional traits for 45 plant species. Original data from \textcite{Caccianiga2006}.}
    \label{fig:rutor_Z}
\end{figure}

\begin{figure}[h!]
    \centering
    \includegraphics{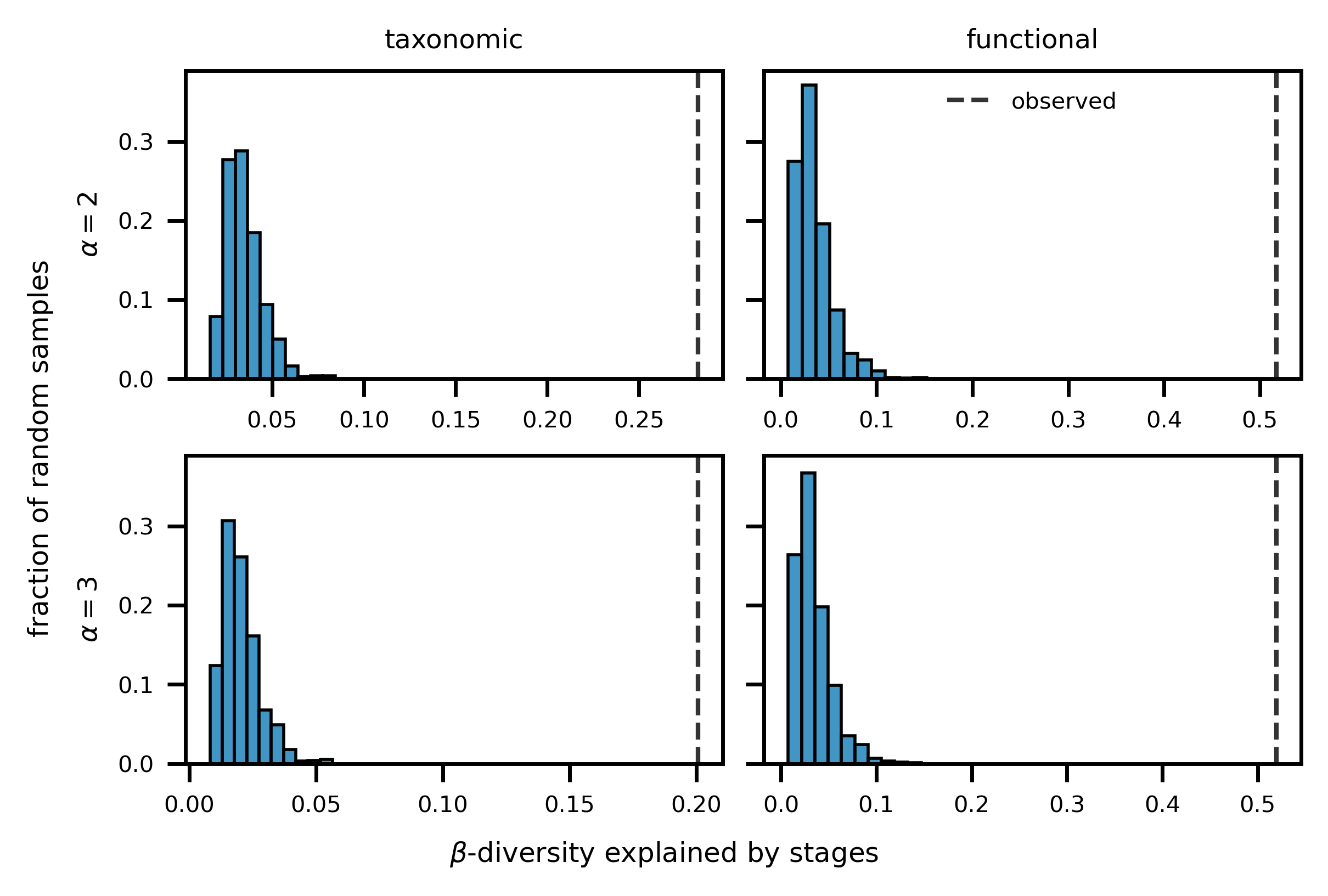}
    \caption{\textbf{Significance of the $\beta$-diversity explained by the stages.} Each subplot compares the observed fraction of $\beta$-diversity explained by the partition into stages to the distribution of expected values under 1,000 drawn from a null model.}
    \label{fig:rutor_explained_null}
\end{figure}

\begin{figure}
    \centering
    \includegraphics{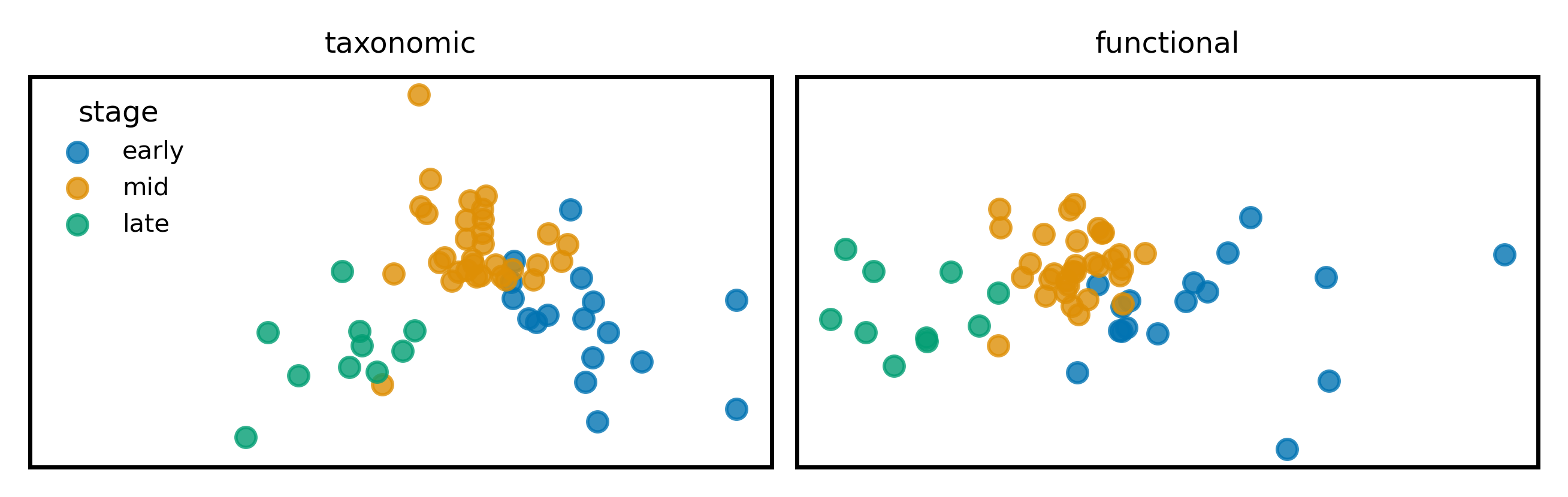}
    \caption{\textbf{Embedding plots based on similarity of species composition.} We compute the Bregman information of each pair of plots and embed them into two dimensions using MDS. Each marker is a plot and similar plots are closer. We fix $\alpha=2$ and consider both taxonomic (left) and functional (right) comparisons. We colour the plots by stage.}
    \label{fig:rutor_plots_embedding}
\end{figure}

\begin{figure}[h!]
    \centering
    \includegraphics{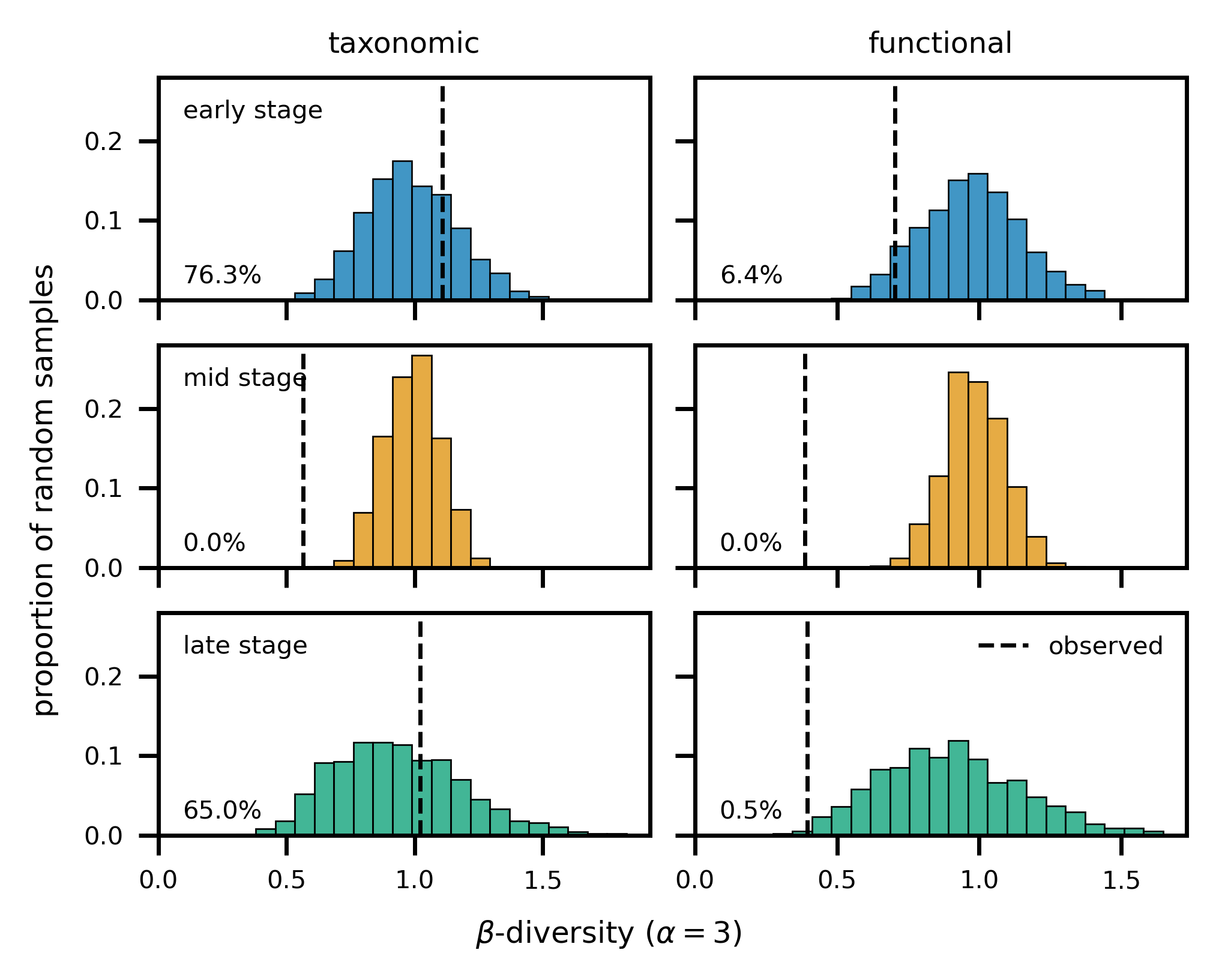}
    \caption{\textbf{$\beta$-diversity of the successional stages of the Rutor glacier.} We plot the taxonomic (left) and functional (right) $\beta$-diversity of the three successional stages (rows). The distributions are the $\beta$-diversities of 1,000 draws from a null model where the partition into stages is random. The dashed black lines mark the empirical values. We annotate the bottom left corner of each subplot with the percentage of samples with $\beta$-diversity lower than the empirical. These results are for $\alpha=3$.}
    \label{fig:rutor_3}
\end{figure}

\end{document}